\documentclass[letterpaper,twocolumn,12pt]{ieee}
\pdfoutput=1

\usepackage{amsfonts}
\usepackage[cmex10]{amsmath}
\usepackage{amssymb}
\usepackage{amsthm}
\usepackage[american]{babel}
\usepackage{cite}
\usepackage{fancyhdr}
\usepackage{fullpage}
\usepackage[paper=letterpaper,centering,margin=1in]{geometry}
\usepackage[pdftex]{graphicx}
\usepackage[colorlinks,frenchlinks,pdftex]{hyperref}
\usepackage{microtype}
\usepackage{relsize}
\usepackage{svn}
\usepackage[nofancy]{svninfo}
\usepackage{times}
\usepackage{xspace}

\SVN $Date: 2009-11-14 10:39:50 -1000 (Sat, 14 Nov 2009) $
\SVN $Revision: 2757 $
\SVN $Id: paper.tex 2757 2009-11-14 20:39:50Z marklee $

\newcommand\versionInfo{\svnToday:Revision~\svnInfoRevision}

\renewcommand\headrule{}
\title{Resource Allocation using Virtual Clusters}

\author{
Mark Stillwell$^{1}$ \and David Schanzenbach$^{1}$ \and
Fr\'ed\'eric Vivien$^{2,3,4,1}$ \and Henri Casanova$^{1}$
~\\
$^1$ Department of Information and Computer Sciences, \\
University of Hawai`i at M\={a}noa, Honolulu, U.S.A.\\
$^2$ INRIA, France\\
$^3$ Universit\'e de Lyon, France\\
$^4$ LIP, UMR 5668 ENS-CNRS-INRIA-UCBL, Lyon, France
}

\hypersetup{
    pdfauthor   = {Mark Stillwell and David Schanzenbach and \
        Fr\'ed\'eric Vivien and Henri Casanova},
    pdftitle    = {Algorithms for Flexible Resource Allocation with Virtual \
        Clusters -- \versionInfo},
    pdfsubject  = {Virtual Cluster Scheduling},
    pdfkeywords = {cluster,scheduler,virtual machine,flexible,np-hard, \
        mixed-integer,linear}
}

\newif\ifremark
\long\def\remark#1{
\ifremark%
        \begingroup%
        \dimen0=\columnwidth
        \advance\dimen0 by -1in%
        \setbox0=\hbox{\parbox[b]{\dimen0}{\protect\em #1}}
        \dimen1=\ht0\advance\dimen1 by 2pt%
        \dimen2=\dp0\advance\dimen2 by 2pt%
        \vskip 0.25pt%
        \hbox to \columnwidth{%
                \vrule height\dimen1 width 3pt depth\dimen2%
                \hss\copy0\hss%
                \vrule height\dimen1 width 3pt depth\dimen2%
        }%
        \endgroup%
\fi}

\newtheorem{theorem}{Theorem}

\begin{document}

\newcommand{\I}{\mathcal{I}}
\newcommand{\probzero}{\textsc{VCSched}\xspace}
\newcommand{\threepart}{\textsc{3-Partition}\xspace}
\newcommand{\twopart}{\textsc{2-Partition}\xspace}
\newcommand{\probzerodec}{\textsc{VCSched-Dec}\xspace}
\newcommand{\N}{\mathbb{N}}
\newcommand{\Q}{\mathbb{Q}}

\remarktrue

\begingroup
\maketitle
\thispagestyle{fancy}
\endgroup

\vspace{-0.3in}
\begin{abstract}
In this report we demonstrate the potential utility of resource allocation
management systems that use virtual machine technology for sharing parallel
computing resources among competing jobs. We formalize the resource allocation
problem with a number of underlying assumptions, determine its complexity,
propose several heuristic algorithms to find near-optimal solutions, and
evaluate these algorithms in simulation. We find that among our algorithms one
is very efficient and also leads to the best resource allocations. We then
describe how our approach can be made more general by removing several of the
underlying assumptions.
\end{abstract}

\section{Introduction}

The use of commodity clusters has become mainstream for high-performance
computing applications, with more than 80\% of today's fastest supercomputers
being clusters~\cite{top500}.  Large-scale data
processing~\cite{map-reduce-osdi04, hadoop, isard-dryad-eurosys07} and service
hosting~\cite{ibm-blue-cloud, amazon-ec2} are also common applications.  These
clusters represent significant equipment and infrastructure investment, and
having a high rate of utilization is key for justifying their ongoing costs
(hardware, power, cooling, staff)~\cite{koomey-server-power07,
epa-server-power07}. There is therefore a strong incentive to share these
clusters among a large number of applications and users.

The sharing of compute resources among competing instances of applications, or
\emph{jobs}, within a single system has been supported by operating systems for
decades via time-sharing.  Time-sharing is implemented with rapid
context-switching and is motivated by a need for interactivity. A fundamental
assumption is that there is no or little a-priori knowledge regarding the
expected workload, including expected durations of running processes.  This is
very different from the current way in which clusters are shared. Typically,
users request some fraction of a cluster for a specified duration. In the
traditional high-performance computing arena, the ubiquitous approach is to use
``batch scheduling'', by which jobs are placed in queues waiting to gain
exclusive access to a subset of the platform for a bounded amount of time. In
service hosting or cloud environments, the approach is to allow users to lease
``virtual slices'' of physical resources, enabled by virtual machine technology.
The latter approach has several advantages, including O/S customization and
interactive execution. In general resource sharing among competing jobs is
difficult because jobs have different resource requirements (amount of
resources, time needed) and because the system cannot accommodate all jobs at
once.

An important observation is that both resource allocation approaches mentioned
above dole out integral subsets of the resources, or \emph{allocations} (e.g.,
10 physical nodes, 20 virtual slices), to jobs. Furthermore, in the case of
batch scheduling, these subsets cannot change throughout application execution.
This is a problem because most applications do not use all resources allocated
to them at all times. It would then be useful to be able to decrease and
increase application allocations on-the-fly (e.g., by removing and adding more
physical cluster nodes or virtual slices during execution). Such application are
termed ``malleable'' in the literature. While solutions have been studied to
implement and to schedule malleable applications~\cite{malleable_1, malleable_2,
malleable_3, malleable_4, malleable_5}, it is often difficult to make sensible
malleability decisions at the application level.  Furthermore, many applications
are used as-is, with no desire or possibility to re-engineer them to be
malleable. As a result sensible and automated malleability is rare in real-world
applications. This is perhaps also due to the fact that production batch
scheduling environments do not provide mechanisms for dynamically increasing or
decreasing allocations. By contrast, in service hosting or cloud environments,
acquiring and relinquishing virtual slices is straightforward and can be
implemented via simple mechanisms. This provides added motivation to engineer
applications to be malleable in those environments. 

Regardless, an application that uses only 80\% of a cluster node or of a virtual
slice would need to relinquish only 20\% of this resources. However, current
resource allocation schemes allocate integral numbers of resources (whether
these are physical cluster nodes or virtual slices). Consequently, many
applications are denied access to resources, or delayed, in spite of cluster
resources not being fully utilized by the applications that are currently
executing, which hinders both application throughput and cluster utilization.

The second limitation of current resource allocation schemes stems from the fact
that resource allocation with integral allocations is difficult from a
theoretical perspective~\cite{bender_soda98}.  Resource allocation problems are
defined formally as the optimizations of well-defined objective functions. Due
to the difficulty (i.e., NP-hardness) of resource allocation for optimizing an
objective function, in the real-world no such objective function is optimized.
For instance, batch schedulers instead provide a myriad of configuration
parameters by which a cluster administrator can tune the scheduling behavior
according to ad-hoc rules of thumb. As a result, it has been noted that there is
a sharp disconnect between the desires of users (low application turn-around
time, fairness) and the schedules computed by batch
schedulers~\cite{schwiegelshohn_sche00, snavely_07}. It turns out that cluster
administrators often attempt to maximize cluster utilization. But recall that,
paradoxically, current resource allocation schemes inherently hinder cluster
utilization!

A notable finding in the theoretical literature is that with job preemption
and/or migration there is more flexibility for resource allocation. In this case
certain resource allocation problems become (more) tractable or
approximable~\cite{baker_1974, schwiegelshohn_sche00, vivien_2008, muthu_1999,
bender_jos2004}. Unfortunately, preemption and migration are rarely used on
production parallel platforms. The gang scheduling~\cite{gang_scheduling}
approach allows entire parallel jobs to be context-switched in a synchronous
fashion.  Unfortunately, a known problem with this approach is the overhead of
coordinated context switching on a parallel platform.  Another problem is the
memory pressure due to the fact that cluster applications often use large
amounts of memory, thus leading to costly swapping between memory and
disk~\cite{gang_scheduling_mem}. Therefore, while flexibility in resource
allocations is desirable for solving resource allocation problems, affording
this flexibility has not been successfully accomplished in production systems.  

In this paper we argue that both limitations of current resource allocation
schemes, namely, reduced utilization and lack of an objective function, can be
addressed simultaneously via fractional and dynamic resource allocations enabled
by state-of-the-art virtual machine (VM) technology.  Indeed, applications
running in VM instances can be monitored so as to discover their resource needs,
and their resource allocations can be modified dynamically (by appropriately
throttling resource consumption and/or by migrating VM instances). Furthermore,
recent VM technology advances make the above possible with low overhead.
Therefore, it is possible to use this technology for resource allocation based
on the optimization of sensible objective functions, e.g., ones that capture
notions of performance and fairness.

Our contributions are:
\begin{itemize}
\item We formalize a general resource allocation problem based on a number of
      assumptions regarding the platform, the workload, and the underlying VM
      technology;
\item We establish the complexity of the problem and propose algorithms to solve
      it;
\item We evaluate our proposed algorithms in simulation and identify an
      algorithm that is very efficient and leads to better resource allocations
      than its competitors;
\item We validate our assumptions regarding the capabilities of VM technology;
\item We discuss how some of our other assumptions can be removed and our
      approach adapted to parallel jobs and dynamic jobs.
\end{itemize}

This paper is organized as follows. In Section~\ref{sec.overview} we define and
formalize our target problem, we list our assumptions for the base problem, and
we establish its NP-hardness. In Section~\ref{sec.algs} we propose algorithms
for solving the base problem and evaluate these algorithms in simulation in
Section~\ref{sec.results}. Sections~\ref{sec.parallel} and~\ref{sec.dynamic}
study the resource sharing problem with relaxed assumptions regarding the nature
of the workload, thereby handling parallel and dynamic workloads. In
Section~\ref{sec.vmtech} we validate our fundamental assumption that VM
technology allows for precise resource sharing. Section~\ref{sec.related}
discusses related work.  Section~\ref{sec.future} discusses future directions.
Finally, Section~\ref{sec.conclusion} concludes the paper with a summary of our
findings.

\section{Flexible Resource Allocation}
\label{sec.overview}

\begin{figure*}
  \centering
  \includegraphics[width=0.90\textwidth]{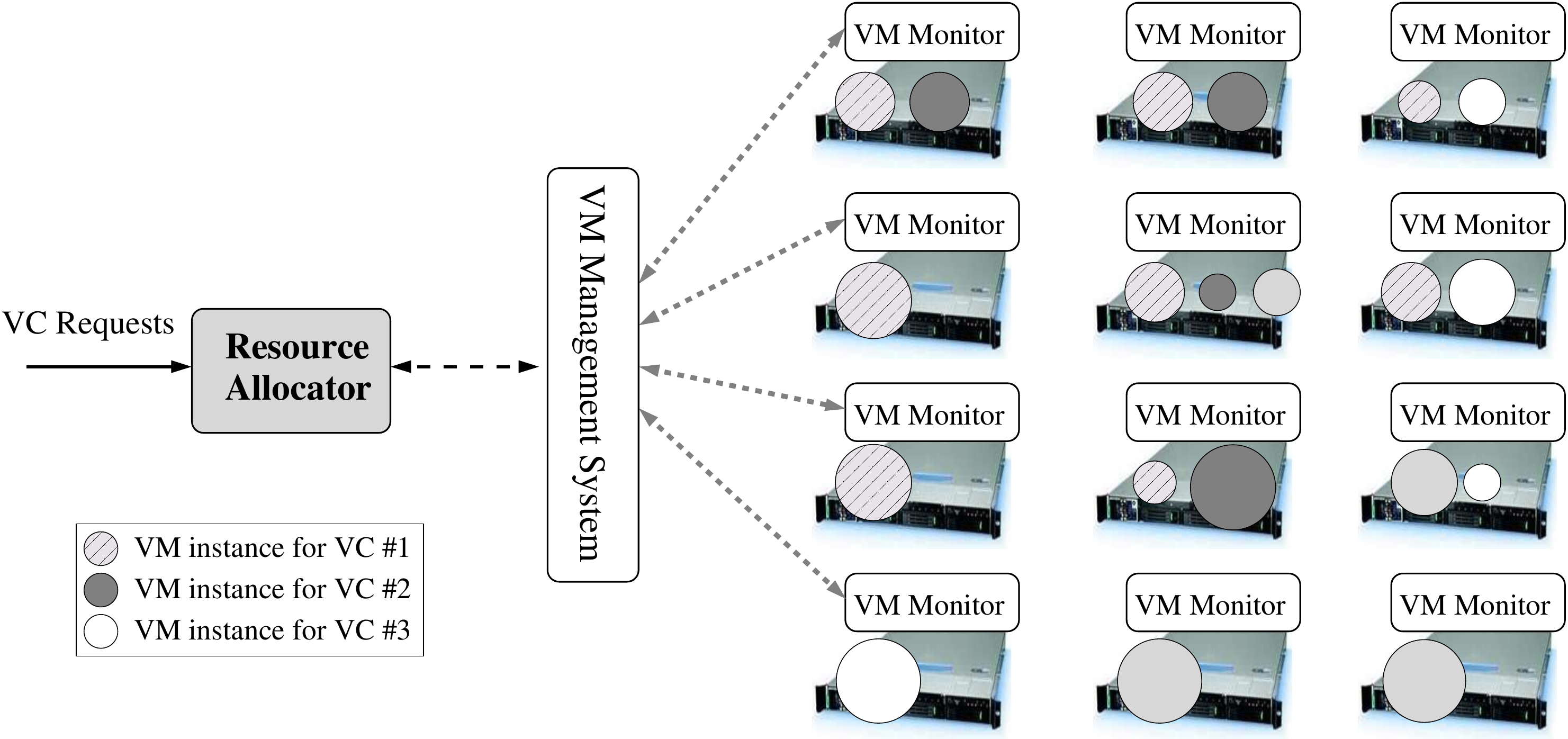}
  \caption{System architecture with 12 homogeneous physical hosts and 3 running
  virtual clusters.}
  \label{fig.system}
\end{figure*}

\subsection{Overview}

In this work we consider a homogeneous cluster platform, which is managed by a
resource allocation system. The architecture of this system is depicted in
Figure~\ref{fig.system}. Users submit job requests, and the system responds by
creating sets of VM instances, or ``virtual clusters'' (VC) to run the jobs.
These instances run on physical hosts that are each under the control of a VM
monitor~\cite{barham03xav, VMWare, Virtual_PC}. The VM monitor can enforce
specific resource consumption rates for different VMs running on the host. All
VM monitors are under the control of a VM management system that can specify
resource consumption rates for VM instances running on the physical cluster.
Furthermore, the VM resource management system can enact VM instance migrations
among physical hosts. An example of such a system is the Usher
project~\cite{Usher}. Finally, a Resource Allocator (RA) makes decisions
regarding whether a request for a VC should be rejected or admitted, regarding
possible VM migrations, and regarding resource consumption rates for each VM
instance.

Our overall goal is to design algorithms implemented as part of the RA that make
all virtual clusters ``play nice'' by allowing fine-grain tuning of their
resource consumptions. The use of VM technology is key for increasing cluster
utilization, as it makes is possible to allocate to VCs only the resources they
need when they need them.  The mechanisms for allowing on-the-fly modification
of resource allocations are implemented as part of the VM Monitors and the VM
Management System.

A difficult question is how to define precisely what ``playing nice'' means, as
it should encompass both notions of individual job performance and notions of
fairness among jobs. We address this issue by defining a performance metric that
encompasses both these notions and that can be used to value resource
allocations. The RA may be configured with the goal of optimizing this metric
but at the same time ensuring that the metric across the jobs is above some
threshold (for instance by rejecting requests for new virtual clusters). More 
generally, a key aspect of our approach is that it can be combined with resource
management and accounting techniques. For instance, it is straightforward to add
notions of user priorities, of resource allocation quotas, of resource
allocation guarantees, or of coordinated resource allocations to VMs belonging
to the same VC. Furthermore, the RA can reject or delay VC requests if the
performance metric is below some acceptable level, to be defined by cluster
administrators.

\subsection{Assumptions}
\label{sec.assumptions}

We first consider the resource sharing problem using the following six
assumptions regarding the workload, the physical platform, and the VM technology
in use:

\begin{description}
\item [(H1)] Jobs are CPU-bound and require a given amount of memory to be able
             to run;
\item [(H2)] Job computational power needs and memory requirements are known;
\item [(H3)] Each job requires only one VM instance;
\item [(H4)] The workload is static, meaning jobs have constant resource
             requirements; furthermore, no job enters or leaves the system; 
\item [(H5)] VM technology allows for precise, low-overhead, and quickly
             adaptable sharing of the computational capabilities of a host
             across CPU-bound VM instances.
\end{description}

These assumptions are very stringent, but provide a good framework to formalize
our resource allocation problem (and to prove that it is difficult even with
these assumptions). We relax assumption H3 in Section~\ref{sec.parallel}, that
is, we consider parallel jobs.  Assumption H4 amounts to assuming that jobs have
no time horizons, i.e., that they run forever with unchanging requirements. In
practice, the resource allocation may need to be modified when the workload
changes (e.g., when a new job arrives, when a job terminates, when a job starts
needing more/fewer resources). In Section~\ref{sec.dynamic} we relax assumption
H4 and extend our approach to allow allocation adaptation. We validate
assumption H5 in Section~\ref{sec.vmtech}. We leave relaxing H1 and H2 for
future work, and discuss the involved challenges in
Section~\ref{sec.conclusion}.

\subsection{Problem Statement}
\label{sec.problem_statement}

We call the resource allocation problem described in the previous section
\probzero and define it here formally. Consider $H>0$ identical physical hosts
and $J>0$ jobs. For job $i$, $i=1,\ldots,J$, let $\alpha_i$ be the (average)
fraction of a host's computational capability utilized by the job if alone on a
physical host, $0 \leq \alpha_i \leq 1$. (Alternatively, this fraction could be
specified a-priori by the user who submitted/launched job $i$.) Let $m_{i}$ be
the maximum fraction of a host's memory needed by job $i$, $0 \leq m_i \leq 1$.
Let $\alpha_{ij}$ be the fraction of the computational capability of host $j$,
$j=1,\ldots,H$, allocated to job $i$, $i=1,\ldots,J$. We have
$0 \leq \alpha_{ij} \leq 1$. If $\alpha_{ij}$ is constrained to be an integer,
that is either $0$ or $1$, then the model is that of scheduling with exclusive
access to resources.  If, instead, $\alpha_{ij}$ is allowed to take rational
values between $0$ and $1$, then resource allocations can be fractional and thus
more fine-grain.

\paragraph{Constraints --}
We can write a few constraints due to resource limitations.  We have
\begin{equation*}
\forall j \quad \sum_{i=1}^{J} \alpha_{ij} \leq 1 \;,
\end{equation*}
which expresses the fact that the total CPU fraction allocated to jobs on any
single host may not exceed 100\%. Also, a job should not be allocated
more resource than it can use:
\begin{equation*}
\forall i \quad \sum_{j=1}^{H} \alpha_{ij} \leq \alpha_{i} \;,
\end{equation*}

\noindent
Similarly,
\begin{equation}
\label{eq.nonlinear1}
\forall j \quad \sum_{i=1}^{J}\lceil\alpha_{ij}\rceil m_i \leq 1 \;,
\end{equation}
since at most the entire memory on a host may be used. 

With assumption H3, a job requires only one VM instance.  Furthermore, as
justified hereafter, we assume that we do not use migration and that a job can
be allocated to a single host. Therefore, we write the following constraints:
\begin{equation}
\label{eq.nonlinear2} 
\forall i \quad \sum_{j=1}^{H} \lceil\alpha_{ij}\rceil = 1\;,
\end{equation} 
which state that for all $i$ only one of the $\alpha_{ij}$ values is non-zero.

\paragraph{Objective function --}
We wish to optimize a performance metric that encompasses both notions of
performance and of fairness, in an attempt at designing the scheduler from the
start with a user-centric metric in mind (unlike, for instance, current batch
schedulers). In the traditional parallel job scheduling literature, the metric
commonly acknowledged as being a good measure for both performance and fairness
is the \emph{stretch} (also called ``slowdown'')~\cite{bender_soda98,
feitelson_slowdown}. The stretch of a job is defined as the job's turn-around
time divided by the turn-around time that would have been achieved had the job
been alone in the system.

This metric cannot be applied directly in our context because jobs have no time
horizons. So, instead, we use a new metric, which we call the \emph{yield} and
which we define for job $i$ as $\sum_{j}\alpha_{ij}/\alpha_{i}$. The yield of a
job represents the fraction of its maximum achievable compute rate that is
achieved (recall that for each $i$ only one of the $\alpha_{ij}$ is non-zero). A
yield of $1$ means that the job consumes compute resources at its peak rate. We
can now define problem \probzero as maximizing the minimum yield in an attempt
at optimizing both performance and fairness (similar in spirit to minimizing the
maximum stretch~\cite{bender_soda98,vivien_2008}). Note that we could easily
maximize the average yield instead, but we may then decrease the fairness of the
resource allocation across jobs as average metrics are
starvation-prone~\cite{vivien_2008}. Our approach is agnostic to the particular
objective function (although some of our results hold only for linear objective
functions). For instance, other ways in which the stretch can be optimized have
been proposed~\cite{Bansal_sigmetric_2001} and could be adapted for our yield
metric.

\paragraph{Migration --}
The formulation of our problem precludes the use of migration. However, as when
optimizing job stretch, migration could be used to achieve better results.
Indeed, assuming that migration can be done with no overhead or cost whatsoever,
migrating tasks among hosts in a periodic steady-state schedule afford more
flexibility for resource sharing, which could in turn be used to maximize the
minimum yield further. For instance, consider 2 hosts and 3 tasks, with
$\alpha_1 = \alpha_2 = \alpha_3 = 0.6$. Without migration the optimal minimum
yield is $0.5/0.6 \sim .83$ (which corresponds to an allocation in which two
tasks are on the same host and each receive 50\% of that host's computing
power). With migration it is possible to do better. Consider a periodic schedule
that switches between two allocations, so that on average the schedule uses each
allocation 50\% of the time. In the first allocation tasks $1$ and $2$ share the
first host, each receiving 45\% and 55\% of the host's computing power, 
respectively, and task $3$ is on the second host by itself, thus receiving 60\%
of its compute power. In the second allocation, the situation is reversed, with
task 1 by itself on the first host and task 2 and 3 on the second host, task 2
receiving 55\% and task 3 receiving 45\%. With this periodic schedule, the
average yield of task 1 and 3 is $.5\times(.45/.60 + .60/.60) \sim .87$ , and
the average yield of task 2 is $.55 / .60 \sim .91$. Therefore the minimum yield
is $.87$, which is higher than that in the no-migration case. 

Unfortunately, the assumption that migration comes at no cost or overhead is not
realistic. While recent advances in VM migration
technology~\cite{xen-migrate-nsdi05} make it possible for a VM instance to
change host with a nearly imperceptible delay, migration consumes network
resources. It is not clear whether the pay-off of these extra migrations would
justify the added cost. It could be interesting to allow a bounded number of
migrations for the purpose of further increasing minimum yield, but for now we
leave this question for future work. We use migration only for the purpose of
adapting to dynamic workloads (see Section~\ref{sec.dynamic}).

\subsection{Complexity Analysis}

Let us consider the decision problem associated to \probzero: Is it possible to
find an allocation so that its minimum yield is above a given bound, $K$?  We
term this problem \probzerodec.  Not surprisingly, \probzerodec is NP-complete.
For instance, considering only job memory constraints and two hosts, the problem
trivially reduces to \twopart, which is known to be NP-complete in the weak
sense~\cite{GareyJohnson}. We can actually prove a stronger result: 

\begin{theorem}
\probzerodec is NP-complete in the strong sense even if host memory capacities
are infinite.
\end{theorem}

\begin{proof}

\probzerodec belongs to NP because a solution can easily be checked in
polynomial time. To prove NP-completeness, we use a straightforward reduction
from \threepart, which is known to be NP-complete in the strong
sense~\cite{GareyJohnson}. Let us consider, $\I_1$, an arbitrary instance of
\threepart: given $3n$ positive integers $\{a_1, \ldots, a_{3n}\}$ and a bound
$R$, assuming that $\frac{R}{4} < a_i < \frac{R}{2}$ for all $i$ and that
$\sum_{i=j}^{3n} a_j = nR$, is there a partition of these numbers into $n$
disjoint subsets $I_1, \ldots, I_n$ such that $\sum_{j\in I_i} a_j = R$ for all
$i$? (Note that $|I_i| = 3$ for all $i$.) We now build $\I_2$, an instance of
\probzero as follows. We consider $H=n$ hosts and $J=3n$ jobs. For job $j$ we
set $\alpha_j = a_j / R$ and $m_j = 0$.  Setting $m_j$ to $0$ amounts to
assuming that there is no memory contention whatsoever, or that host memories
are infinite. Finally, we set $K$, the bound on the yield, to be $1$.  We now
prove that $\I_1$ has a solution if and only if $\I_2$ has a solution.  

Let us assume that $\I_1$ has a solution. For each job $j$, we assign it to host
$i$ if $j \in I_i$, and we give it all the compute power it needs
($\alpha_{ji} = a_j / R$). This is possible because $\sum_{j\in I_i} a_j = R$,
which implies that $\sum_{j\in I_i} \alpha_{ji} = R/R \leq 1$. In other terms,
the computational capacity of each host is not exceeded. As a result, each job
has a yield of $K=1$ and we have built a solution to $\I_2$.  

Let us now assume that $\I_2$ has a solution. Then, for each job $j$ there
exists a unique $i_j$ such that $\alpha_{ji_j} = \alpha_j$, and such that
$\alpha_{ji} = 0$ for $i \neq i_j$ (i.e., job $j$ is allocated to host $i_j$).
Let us define $I_i = \{j | i_j = i\}$. By this definition, the $I_i$ sets are
disjoint and form a partition of $\{1,\ldots,3n\}$.  

To ensure that each processor's compute capability is not exceeded, we must have
$\sum_{j \in I_i} \alpha_j \leq 1$ for all $i$. However, by construction of
$\I_2$, $\sum_{j=1}^{3n} \alpha_j = n$. Therefore, since the $I_i$ sets form a
partition of $\{1,\ldots,3n\}$, $\sum_{j \in I_i} \alpha_j$ is exactly equal to
1 for all $i$. Indeed, if $\sum_{j \in I_{i_1}}\alpha_j$ were strictly lower
than $1$ for some $i_1$, then $\sum_{j \in I_{i_2}}\alpha_j$ would have to be
greater than $1$ for some $i_2$, meaning that the computational capability of a
host would be exceeded. Since $\alpha_j = a_j / R$, we obtain
$\sum_{j \in I_i} a_j = R$ for all $i$. Sets $I_i$ are thus a solution to
$\I_1$, which concludes the proof.

\end{proof}

\subsection{Mixed-Integer Linear Program Formulation}
\label{sec.milp}

It turns out that \probzero can be formulated as a mixed-integer linear program
(MILP), that is an optimization problem with linear constraints and objective
function, but with both rational and integer variables. Among the constraints
given in Section~\ref{sec.problem_statement}, the constraints in
Eq.~\ref{eq.nonlinear1} and Eq.~\ref{eq.nonlinear2} are non-linear. These
constraints can easily be made linear by introducing a binary integer variables,
$e_{ij}$, set to 1 if job $i$ is allocated to resource $j$, and to 0 otherwise.
We can then rewrite the constraints in Section~\ref{sec.problem_statement} as
follows, with $i = 1,\ldots, J$ and $j=1, \ldots, H$:
\begin{eqnarray}
\label{eq.constraints}
\forall i,j & \quad e_{ij} \in \N,\\
\forall i,j & \quad \alpha_{ij} \in \Q,\\
\forall i,j & 0 \leq e_{ij} \leq 1,\\
\forall i,j & 0 \leq \alpha_{ij} \leq e_{ij},\\
\forall i   & \sum_{j=1}^{H} e_{ij} = 1,\\
\forall j   & \sum_{i=1}^{J} \alpha_{ij} \leq 1,\\
\forall j   & \sum_{i=1}^{J} e_{ij} m_{i} \leq 1 \;\\
\forall i   & \sum_{j=1}^{H} \alpha_{ij} \leq \alpha_i \;\\
\forall i   & \sum_{j=1}^{H} \frac{\alpha_{ij}}{\alpha_{i}} \geq Y
\end{eqnarray}
Recall that $m_i$ and $\alpha_i$ are constants that define the jobs. The
objective is to maximize $Y$, i.e., to maximize the minimum yield.

\section{Algorithms for Solving \probzero}
\label{sec.algs}

In this section we propose algorithms to solve \probzero, including exact and
relaxed solutions of the MILP in Section~\ref{sec.milp} as well as ad-hoc
heuristics. We also give a generally applicable technique to improve average
yield further once the minimum yield has been maximized.

\subsection{Exact and Relaxed Solutions}

In general, solving a MILP requires exponential time and is only feasible for
small problem instances. We use a publicly available MILP solver, the Gnu Linear
Programming Toolkit (GLPK), to compute the exact solution when the problem
instance is small (i.e., few tasks and/or few hosts). We can also solve a
relaxation of the MILP by assuming that all variables are rational, converting
the problem to a LP.  In practice a rational linear program can be solved in
polynomial time. However, the resulting solution may be infeasible (namely 
because it could spread a single job over multiple hosts due to non-binary
$e_{ij}$ values), but has two important uses. First, the value of the objective
function is an upper bound on what is achievable in practice, which is useful to
evaluate the absolute performance of heuristics. Second, the rational solution
may point the way toward a good feasible solution that is computed by rounding
off the $e_{ij}$ values to integer values judiciously, as discussed in the next
section.

It turns out that we do not need a linear program solver to compute the optimal
minimum yield for the relaxed program. Indeed, if the total of job memory
requirement is not larger than the total available memory (i.e., if
$\sum_{i=1}^{J} m_i \leq H$), then there is a solution to the relaxed version of
the problem and the achieved optimal minimum yield,
$Y^{(\text{rat})}_{\text{opt}}$, can be computed easily:
\begin{equation*}
Y^{(\text{rat})}_{\text{opt}} = min\{\frac{H}{\sum_{i=1}^J \alpha_i}, 1\}.
\end{equation*}
The above expression is an obvious upper bound on the maximum minimum yield. To
show that it is in fact the optimal, we simply need to exhibit an allocation
that achieves this objective. A simple such allocation is: 
\begin{equation*}
\forall i,j \hspace*{0.1in}
e_{ij}=\frac{1}{H}\hspace*{0.1in} \text{and}
\hspace*{0.1in}\alpha_{ij}=\frac{1}{H}\alpha_i Y^{(\text{rat})}_{\text{opt}}.  
\end{equation*}

\subsection{Algorithms Based on Relaxed Solutions}

We propose two heuristics, RRND and RRNZ, that use a solution of the rational LP
as a basis and then round-off rational $e_{ij}$ value to attempt to produce a
feasible solution, which is a classical technique. In the previous section we
have shown a solution for the LP; Unfortunately, that solution has the
undesirable property that it splits each job evenly across all hosts, meaning 
that all $e_{ij}$ values are identical. Therefore it is a poor starting point
for heuristics that attempt to round off $e_{ij}$ values based on their
magnitude. Therefore, we use GLPK to solve the relaxed MILP and use the produced
solution as a starting point instead.

\noindent
{\bf Randomized Rounding (RRND) --} This heuristic first solves the LP. Then,
for each job $i$ (taken in an arbitrary order), it allocates it to a random host
using a probability distribution by which host $j$ has probability $e_{ij}$ of
being selected.  If the job cannot fit on the selected host because of memory
constraints, then that host's probability of being selected is set to zero and
another attempt is made with the relative probabilities of the remaining hosts
adjusted accordingly. If no selection has been made and every host has zero
probability of being selected, then the heuristic fails. Such a probabilistic
approach for rounding rational variable values into integer values has been used
successfully in previous work~\cite{marchal2006sss}.

\noindent{\bf Randomized Rounding with No Zero probability (RRNZ) --} This
heuristic is a slight modification of the RRND heuristic. One problem with RRND
is that a job, $i$, may not fit (in terms of memory requirements) on any of the
hosts, $j$, for which $e_{ij} > 0$, in which case the heuristic would fail to
generate a solution. To remedy this problem, we set each $e_{ij}$ value equal to
zero in the solution of the relaxed MILP to $\epsilon$ instead, where
$\epsilon << 1$ (we used $\epsilon = 0.01$).  For those problem instances for
which RRND provides a solution RRNZ should provide nearly the same solution most
of the time. But RRNZ should also provide a solution to a some instances for
which RRND fails, thus achieving a better success rate.

\subsection{Greedy Algorithms}

\noindent{\bf Greedy (GR) --} This heuristic first goes through the list of jobs
in arbitrary order. For each job the heuristic ranks the hosts according to
their total computational load, that is, the total of the maximum computation
requirements of all jobs already assigned to a host. The heuristic then selects
the first host, in non-decreasing order of computational load, for which an
assignment of the current job to that host will satisfy the job's memory
requirements.

\noindent{\bf Sorted-Task Greedy (SG) --}  This version of the greedy heuristic
first sorts the jobs in descending order by their memory requirements before
proceeding as in the standard greedy algorithm. The idea is to place relatively
large jobs while the system is still lightly loaded.

\noindent{\bf Greedy with Backtracking (GB) --}  It is possible to modify the GR
heuristic to add backtracking. Clearly full-fledged backtracking would lead to 
100\% success rate for all instances that have a feasible solution, but it would
also require potentially exponential time. One thus needs methods to prune the
search tree. We use a simple method, placing an arbitrary bound (500,000) on the
number of job placement attempts. An alternate pruning technique would be to
restrict placement attempts to the top 25\% candidate placements, but based on
our experiments it is vastly inferior to using an arbitrary bound on job
placement attempts.

\noindent{\bf Sorted Greedy with Backtracking (SGB) --}  This version is a
combination of SG and GB, i.e., tasks are sorted in descending order of memory
requirement as in SG and backtracking is used as in GB.  

\subsection{Multi-Capacity Bin Packing Algorithms}

Resource allocation problems are often akin to bin packing problems, and
\probzero is no exception. There are however two important differences between
our problem and bin packing. First, our tasks resource requirements are dual,
with both memory and CPU requirements. Second, our CPU requirements are not
fixed but depend on the achieved yield. The first difference can be addressed by
using ``multi-capacity'' bin packing heuristics. Two Multi-capacity bin packing
heuristics were proposed in~\cite{leinberger1999mcb} for the general case of
$d$-capacity bins and items, but in the $d=2$ case these two algorithms turn out
to be equivalent. The second difference can be addressed via a binary search on
the yield value. 

Consider an instance of \probzero and a fixed value of the yield, $Y$, that
needs to be achieved. By fixing $Y$, each task has both a fixed memory
requirement and a fixed CPU requirement, both taking values between 0 and 1,
making it possible to apply the algorithm in~\cite{leinberger1999mcb} directly.

Accordingly, one splits the tasks into two lists, with one list containing the
tasks with higher CPU requirements than memory requirements and the other
containing the tasks with higher memory requirements than CPU requirements. One
then sorts each list. In~\cite{leinberger1999mcb} the lists are sorted according
to the sum of the CPU and memory requirements.

Once the lists are sorted, one can start assigning tasks to the first host.
Lists are always scanned in order, searching for a task that can ``fit'' on the
host, which for the sake of this discussion we term a ``possible task''.
Initially one searches for a possible task in one and then the other list,
starting arbitrarily with any list. This task is then assigned to the host.
Subsequently, one always searches for a possible task from the list that goes
against the current imbalance. For instance, say that the host's available
memory capacity is 50\% and its available CPU capacity is 80\%, based on tasks
that have been assigned to it so far. In this case one would scan the list of
tasks that have higher CPU requirements than memory requirements to find a
possible task. If no such possible task is found, then one scans the other list
to find a possible task. When no possible tasks are found in either list, one
starts this process again for the second host, and so on for all hosts. If all
tasks can be assigned in this manner on the available hosts, then resource
allocation is successful.  Otherwise resource allocation fails.

The final yield must be between $0$, representing failure, and the smaller of
$1$ or the total computation capacity of all the hosts divided by the total 
computational requirements of all the tasks. We arbitrarily choose to start at
one-half of this value and perform a binary search of possible minimum yield
values, seeking to maximize minimum yield. Note that under some circumstances
the algorithm may fail to find a valid allocation at a given potential yield
value, even though it would find one given a larger yield value. This type of
failure condition is to be expected when applying heuristics.

While the algorithm in~\cite{leinberger1999mcb} sorts each list by the sum of
the memory and CPU requirements, there are other likely sorting key candidates.
For completeness we experiment with 8 different options for sorting the lists,
each resulting in a MCB (Multi-Capacity Bin Packing) algorithm. We describe all
8 options below:

\begin{itemize}
\item MCB1: memory + CPU, in ascending order;
\item MCB2: max(memory,CPU) - min(memory,CPU), in ascending order;
\item MCB3: max(memory,CPU) / min(memory,CPU), in ascending order;
\item MCB4: max(memory,CPU), in ascending order;
\item MCB5: memory + CPU, in descending order;
\item MCB6: max(memory,CPU) - min(memory,CPU), in descending order.
\item MCB7: max(memory,CPU) / min(memory,CPU), in descending order;
\item MCB8: max(memory,CPU), in descending order;
\end{itemize}

\subsection{Increasing Average Yield}
\label{sec.aveyield}

While the objective function to be maximized for solving \probzero is the
minimum yield, once an allocation that achieves this goal has been found there
may be excess computational resources available which would be wasted if not
allocated. Let us call $\cal{Y}$ the maximized minimum yield value computed by
one of the aforementioned algorithms (either an exact value obtained by solving
the MILP, or a likely sub-optimal value obtained with a heuristic). One can then
solve a new linear program simply by adding the constraint $Y \geq \cal{Y}$ and
seeking to maximize $\sum_{j} \alpha_{ij} / \alpha_{i}$, i.e., the average
yield. Unfortunately this new program also contains both integer and rational
variables, therefore requiring exponential time for computing an exact solution.
Therefore, we choose to impose the additional condition that the $e_{ij}$ values
be unchanged in this second round of optimization. In other terms, only CPU
fractions can be modified to improve average yield, not job locations. This
amounts to replacing the $e_{ij}$ variables by their values as constants when
maximizing the average yield and the new linear program has then only rational
variables.

It turns out that, rather than solving this linear program with a linear program
solver, we can use the following optimal greedy algorithm. First, for each job
$i$ assigned to host $j$, we set the fraction of the compute capability of host
$j$ given to job $i$ to the value exactly achieving the maximum minimum yield:
$\alpha_{ij}=\alpha_i.\cal{Y}$. Then, for each host, we scale up the compute
fraction of the job with smallest compute requirement $\alpha_i$ until either
the host has no compute capability left or the job's compute requirement is
fully fulfilled. In the latter case, we then apply the same scheme to the job
with the second smallest compute requirement on that host, and so on. The
optimality of this process is easily proved via a typical exchange argument.

All our heuristics use this average yield maximization technique after
maximizing the minimum yield.

\section{Simulation Experiments}
\label{sec.results}

We evaluate our heuristics based on four metrics: (i)~the achieved minimum
yield; (ii)~the achieved average yield; (iii)~the failure rate; and (iv)~the run
time. We also compare the heuristics with the exact solution of the MILP for
small instances, and to the (unreachable upper bound) solution of the rational 
LP for all instances. The achieved minimum and average yields considered are
average values over successfully solved problem instances.  The run times given
include only the time required for the given heuristic since all algorithms use
the same average yield maximization technique.

\subsection{Experimental Methodology}
\label{sec.methodology}

We conducted simulations on synthetic problem instances. We defined these
instances based on  the number of hosts, the number of jobs, the total amount of
free memory, or \emph{memory slack}, in the system, the average job CPU
requirement, and the coefficient of variance of both the memory and CPU
requirements of jobs.  The memory slack is used rather than the average job
memory requirement since it gives a better sense of how tightly packed the
system is as a whole. In general (but not always) the greater the slack the
greater the number of feasible solutions to \probzero.

Per-task CPU and memory requirements are sampled from a normal distribution with
given mean and coefficient of variance, truncated so that values are between 0
and 1. The mean memory requirement is defined as $H * (1 - slack) / J$, where
$slack$ has value between 0 and 1. The mean CPU requirement is taken to be 0.5,
which in practice means that feasible instances with fewer than twice as many
tasks as hosts have a maximum minimum yield of 1.0 with high probability. We do
not ensure that every problem instance has a feasible solution.

Two different sets of problem instances are examined.  The first set of
instances, ``small'' problems, includes instances with small numbers of hosts
and tasks. Exact optimal solutions to these problems can be found with a MILP
solver in a tractable amount of time (from a few minutes to a few hours on a
3.2Ghz machine using the GLPK solver). The second set of instances, ``large''
problems, includes instances for which the numbers of hosts and tasks are too
large to compute exact solutions.  For the small problem set we consider 4 hosts
with 6, 8, 10, or 12 tasks. Slack ranges from 0.1 to 0.9 with increments of 0.1,
while coefficients of variance for memory and CPU requirements are given values
of 0.25 and 0.75, for a total of 144 different problem specifications. 10
instances are generated for each problem specification, for a total of 1,440
instances. For the large problem set we consider 64 hosts with sets of 100, 250
and 500 tasks. Slack and coefficients of variance for memory and CPU
requirements are the same as for the small problem set for a total of 108
different problems specifications. 100 instances of each problem specification
were generated for a total of 10,800 instances.

\subsection{Experimental Results}

\subsubsection{Multi-Capacity Bin Packing}

\begin{figure}
  \centering
  \includegraphics[width=0.45\textwidth]{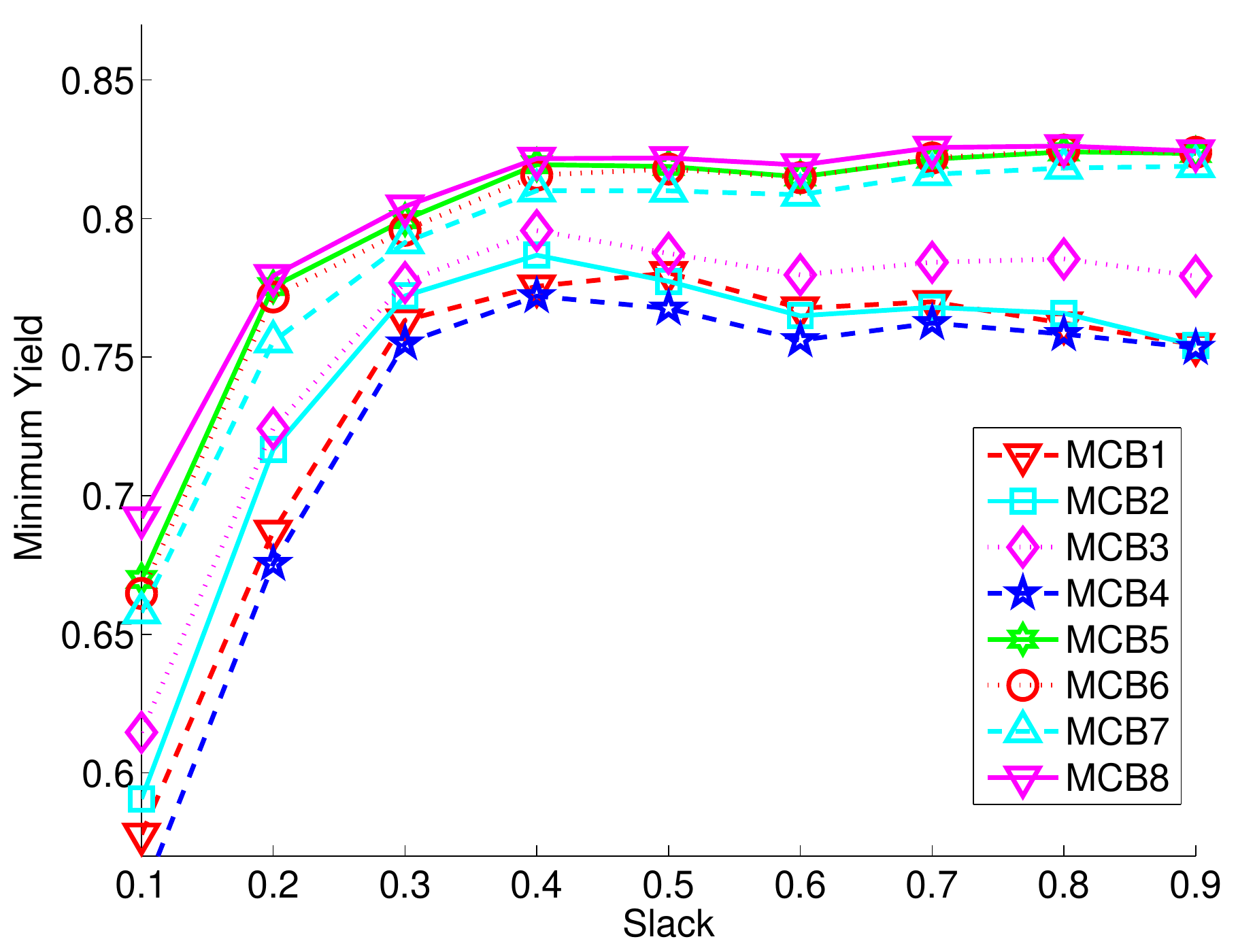}
  \caption{MCB Algorithms -- Minimum Yield vs. Slack for small problem
     instances.}
  \label{fig.mcb-minyield-vs-slack-for-4-hosts}
\end{figure}

\begin{figure}
  \centering
  \includegraphics[width=0.45\textwidth]{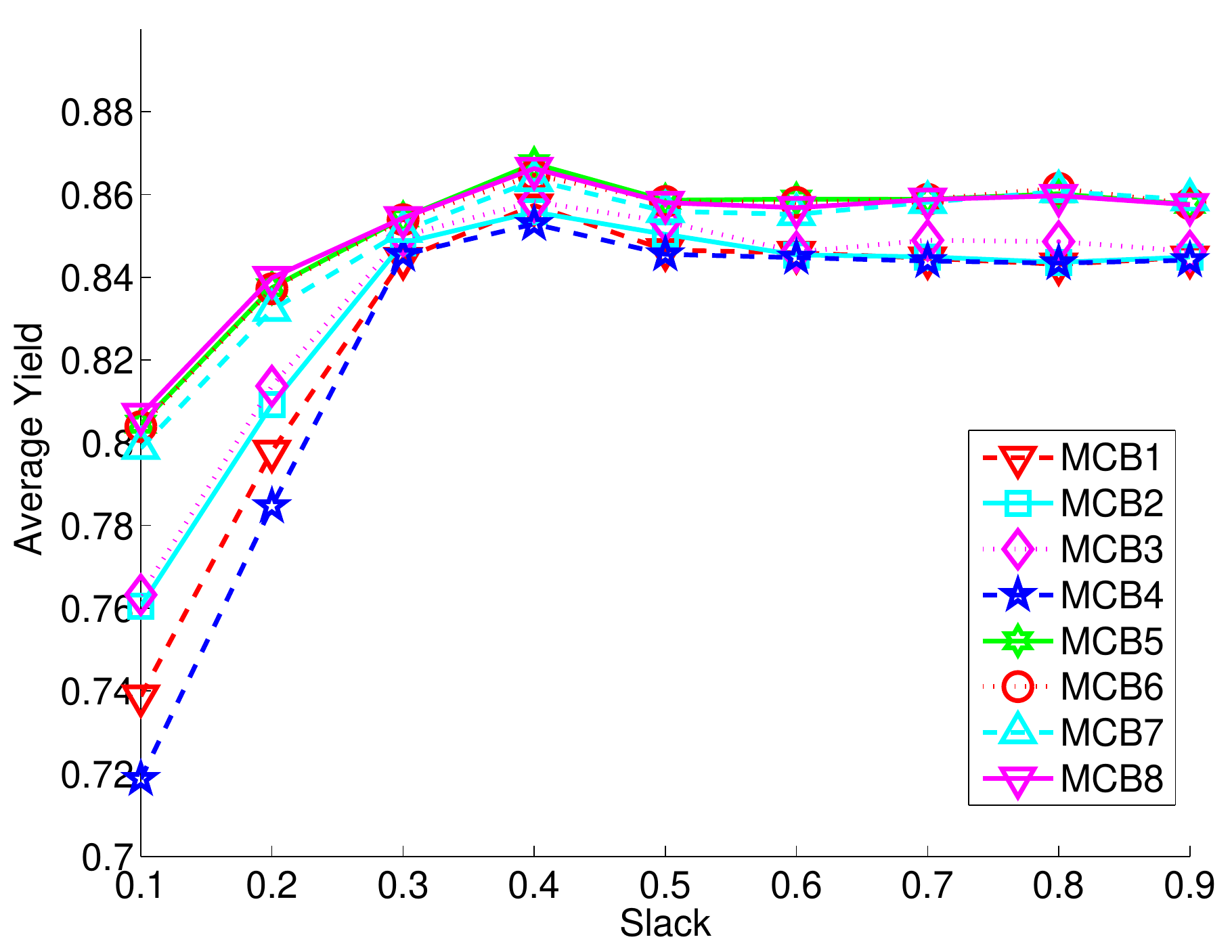}
  \caption{MCB Algorithms -- Average Yield vs. Slack for small problem
     instances.}
  \label{fig.mcb-avgyield-vs-slack-for-4-hosts}
\end{figure}

\begin{table}
\begin{center}
\begin{tabular}{|l|c|c|}
\hline
& \multicolumn{2}{|c|}{\% deg. from best} \\
\cline{2-3}
algorithm & avg. & max \\
\hline
MCB8 &  1.06 & 40.45 \\
MCB5 &  1.60 & 38.67 \\
MCB6 &  1.83 & 37.61 \\
MCB7 &  3.91 & 40.71 \\
MCB3 & 11.76 & 55.73 \\
MCB2 & 14.21 & 48.30 \\
MCB1 & 14.90 & 55.84 \\
MCB4 & 17.32 & 46.95 \\
\hline
\end{tabular}
\end{center}
\caption{Average and Maximum percent degradation from best of the MCB algorithms
for small problem instances.}
\label{tab.mcb-degredation-for-4-hosts}
\end{table}

We first present results only for our 8 multi-capacity bin packing algorithms to
determine the best one. Figure~\ref{fig.mcb-minyield-vs-slack-for-4-hosts} shows
the achieved maximum minimum yield versus the memory slack averaged over small
problem instances. As expected, as the memory slack increases all algorithms
tend to do better although some algorithms seem to experience slight decreases
in performance beyond a slack of $0.4$. Also expected, we see that the four
algorithms that sort the tasks by descending order outperform the four that sort
them by ascending order. Indeed, it is known that for bin packing starting with
large items typically leads to better results on average.

The main message here is that MCB8 outperforms all other algorithms across the
board. This is better seen in Table~\ref{tab.mcb-degredation-for-4-hosts}, which
shows the average and maximum percent degradation from best for all algorithms.
For a problem instance, the percent degradation from best of an algorithm is
defined as the difference, in percentage, between the minimum yield achieved by
an algorithm and the minimum yield achieved by the best algorithm for this
instance. The average and maximum percent degradations from best are computed
over all instances. We see that MCB8 has the lowest average percent degradation
from best. MCB5, which corresponds to the algorithm in~\cite{leinberger1999mcb}
performs well but not as well as MCB8. In terms of maximum percent degradation
from best, we see that MCB8 ranks third, overtaken by MCB5 and MCB6. Examining
the results in more details shows that, for these small problem instances, the
maximum degradation from best are due to outliers. For instance, for the MCB8
algorithm, out of the 1,379 solved instances, there are only 155 instances for
which the degradation from best if larger than 3\%, and only 19 for which it is
larger than 10\%. 

Figure~\ref{fig.mcb-avgyield-vs-slack-for-4-hosts} shows the average yield
versus the slack (recall that the average yield is optimized in a second phase,
as described in Section~\ref{sec.aveyield}). We see here again that the MCB8
algorithm is among the very best algorithms. 

\begin{figure}
  \centering
  \includegraphics[width=0.45\textwidth]{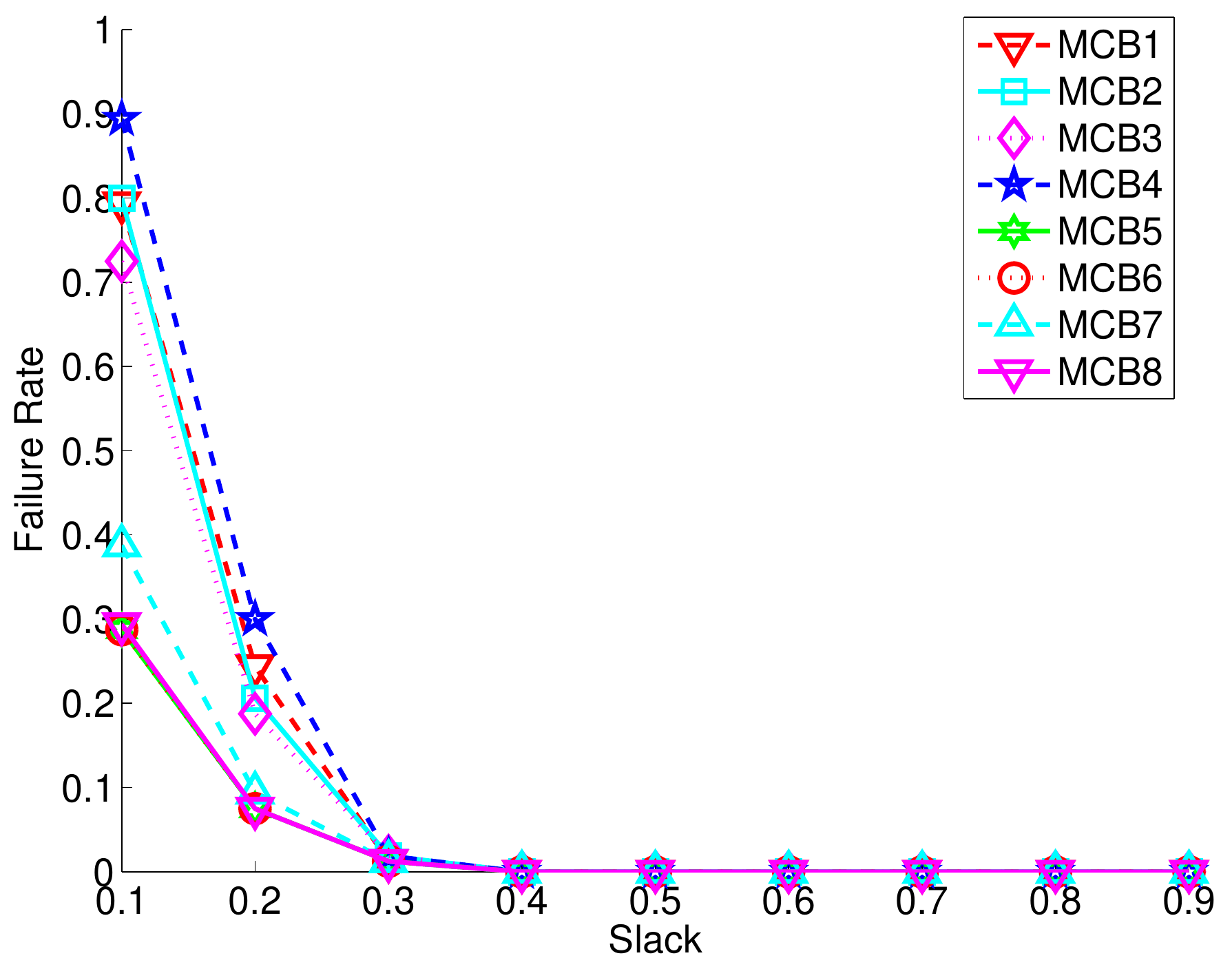}
  \caption{MCB Algorithms -- Failure Rate vs. Slack for small problem
     instances.}
  \label{fig.mcb-failrate-vs-slack-for-4-hosts}
\end{figure}
\begin{figure}
  \centering
  \includegraphics[width=0.45\textwidth]{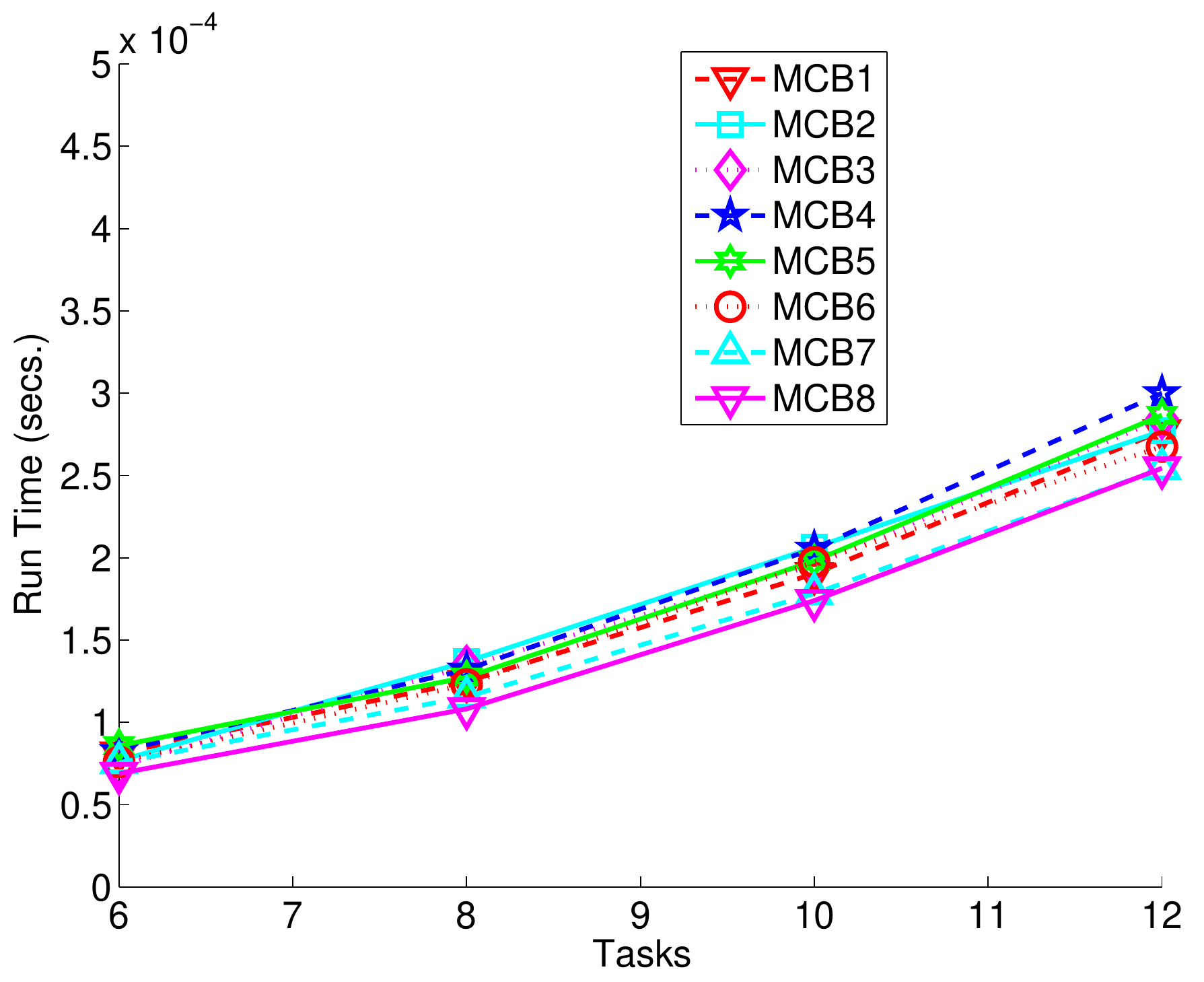}
  \caption{MCB Algorithms -- Run time vs. Number of Tasks for small problems
     instances.}
  \label{fig.mcb-runtime-vs-tasks-for-4-hosts}
\end{figure}

Figure~\ref{fig.mcb-failrate-vs-slack-for-4-hosts} shows the failure rates of
the 8 algorithms versus the memory slack. As expected failure rates decrease as
the memory slack increases, and as before we see that the four algorithms that
sort tasks by descending order outperform the algorithms that sort tasks by
ascending order. Finally, Figure~\ref{fig.mcb-runtime-vs-tasks-for-4-hosts} 
shows the runtime of the algorithms versus the number of tasks. We use a 3.2GHz
Intel Xeon processor. All algorithms have average run times under 0.18
milliseconds, with MCB8 the fastest by a tiny margin.

\begin{figure}
  \centering
  \includegraphics[width=0.45\textwidth]{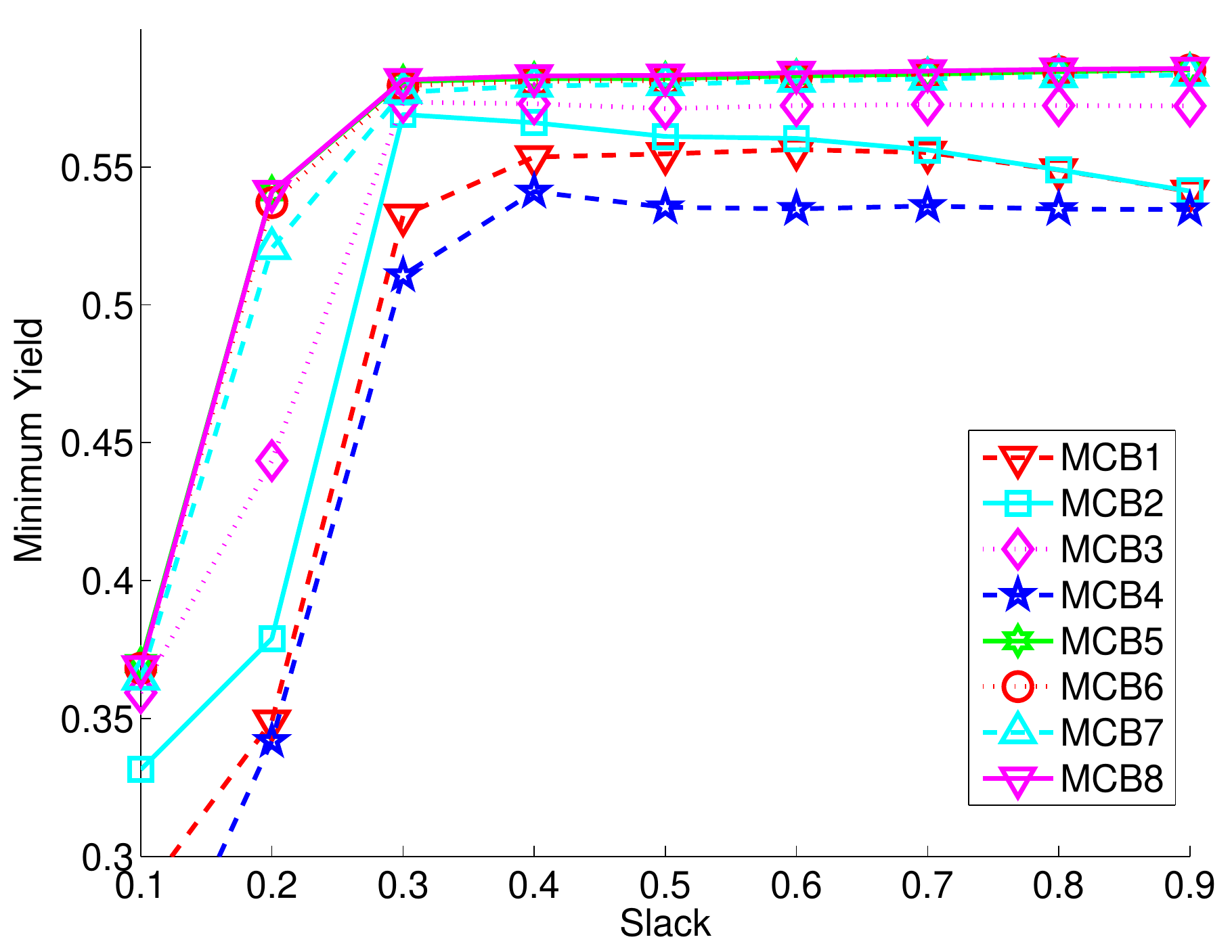}
  \caption{MCB Algorithms -- Minimum Yield vs. Slack for large problem
     instances.}
  \label{fig.mcb-minyield-vs-slack-for-64-hosts}
\end{figure}

\begin{figure}
\begin{center}
\begin{tabular}{|l|c|c|}
\hline
& \multicolumn{2}{|c|}{\% deg. from best} \\
\cline{2-3}
algorithm & avg. & max \\
\hline
MCB8 &  0.09 & 3.16\\
MCB5 &  0.25 & 3.50\\
MCB6 &  0.46 & 16.68\\
MCB7 &  1.04 & 48.39\\
MCB3 &  4.07 & 64.71\\
MCB2 &  8.68 & 46.68\\
MCB1 & 10.97 & 73.33\\
MCB4 & 14.80 & 61.20\\
\hline
\end{tabular}
\end{center}
\caption{Average and Maximum percent degradation from best of the MCB algorithms
for large problem instances.}
\label{tab.mcb-degredation-for-64-hosts}
\end{figure}

\begin{figure}
  \centering
  \includegraphics[width=0.45\textwidth]{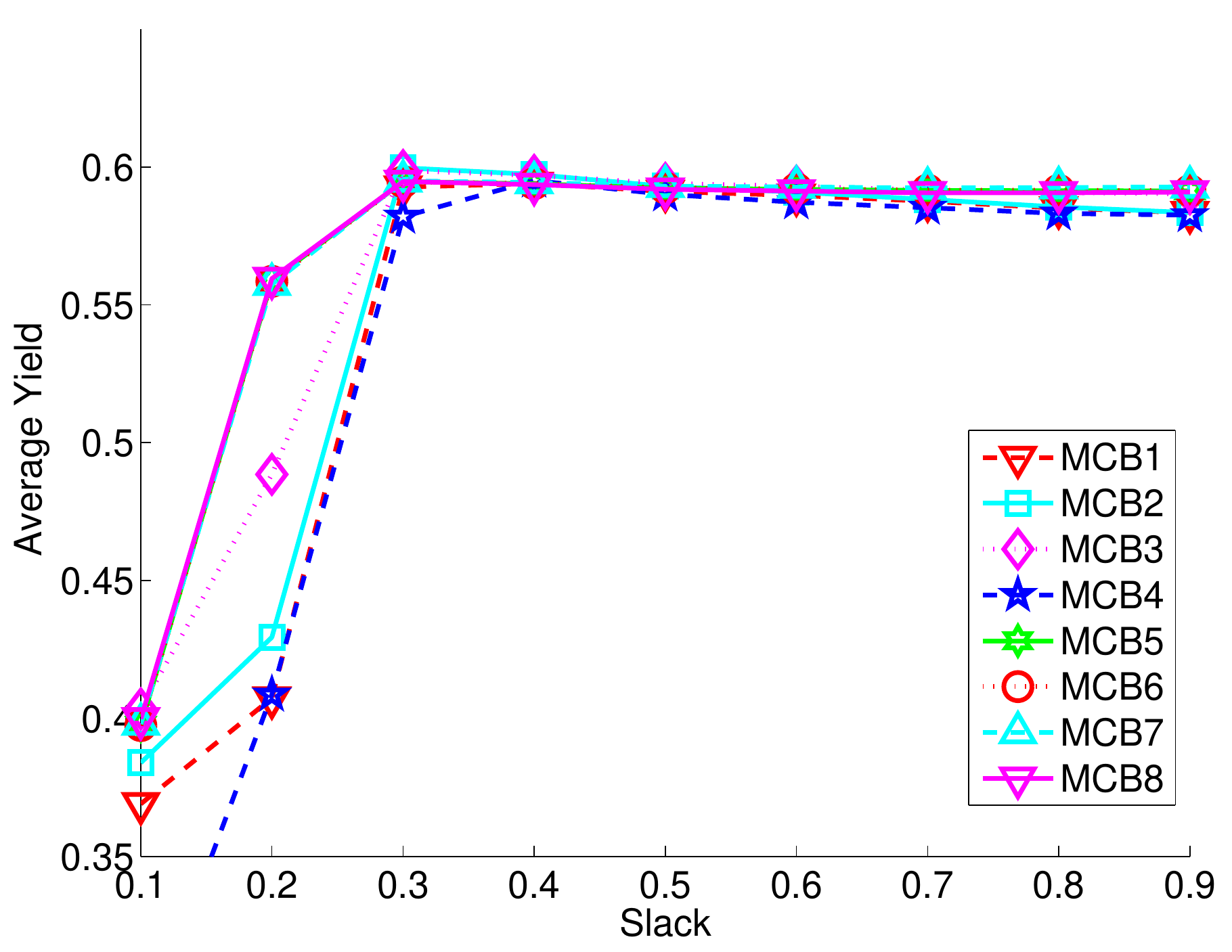}
  \caption{MCB Algorithms -- Average Yield vs. Slack for large problem
     instances.}
  \label{fig.mcb-avgyield-vs-slack-for-64-hosts}
\end{figure}

\begin{figure}
  \centering
  \includegraphics[width=0.45\textwidth]{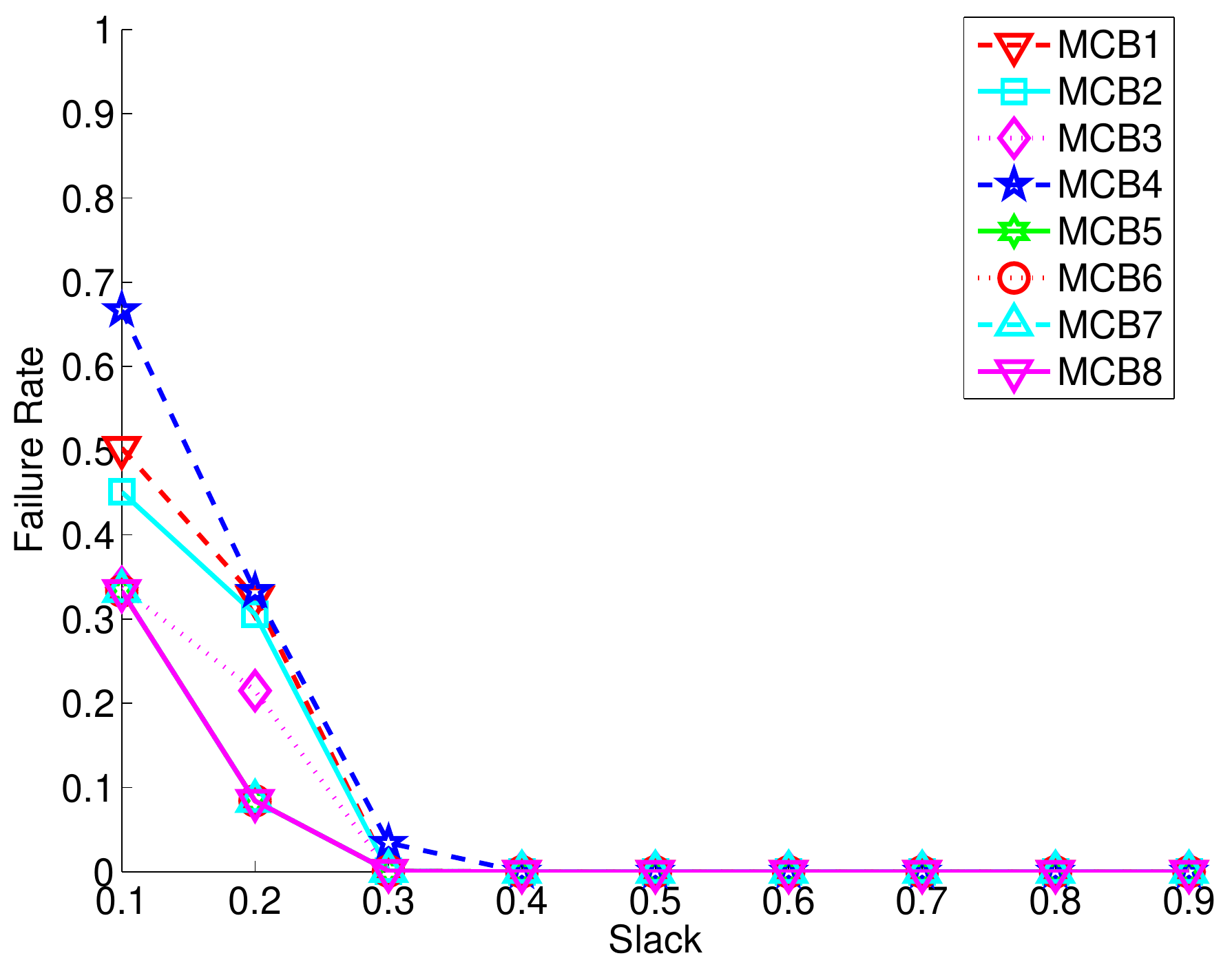}
  \caption{MCB Algorithms -- Failure Rate vs. Slack for large problem
     instances.}
  \label{fig.mcb-failrate-vs-slack-for-64-hosts}
\end{figure}

\begin{figure}
  \centering
  \includegraphics[width=0.45\textwidth]{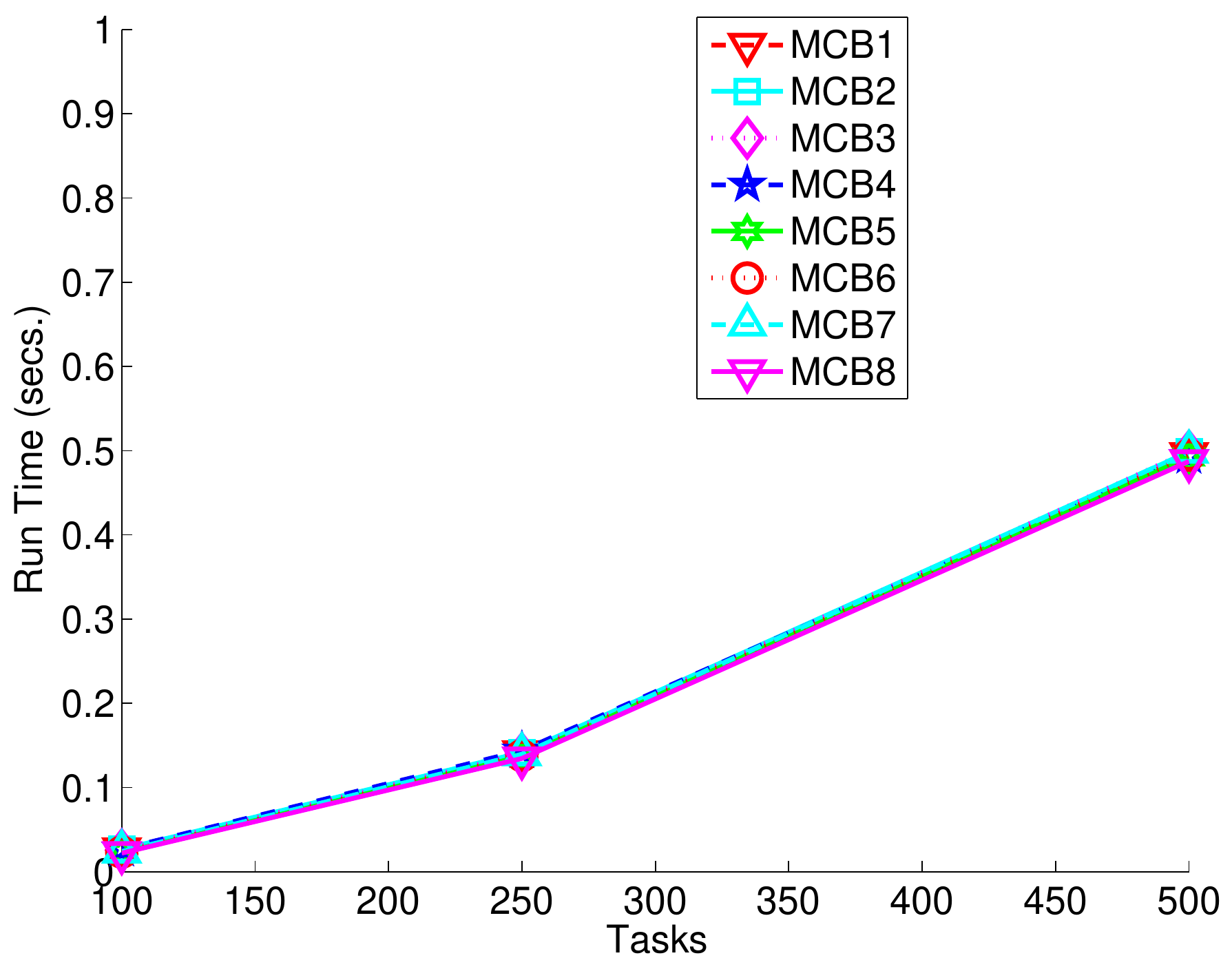}
  \caption{MCB Algorithms -- Run time vs. Number of Tasks for large problem
     instances.}
  \label{fig.mcb-runtime-vs-tasks-for-64-hosts}
\end{figure}

Figures~\ref{fig.mcb-minyield-vs-slack-for-64-hosts},
~\ref{fig.mcb-avgyield-vs-slack-for-64-hosts},
~\ref{fig.mcb-failrate-vs-slack-for-64-hosts},
and~\ref{fig.mcb-runtime-vs-tasks-for-64-hosts} are similar to
Figures~\ref{fig.mcb-minyield-vs-slack-for-4-hosts},
~\ref{fig.mcb-avgyield-vs-slack-for-4-hosts},
~\ref{fig.mcb-failrate-vs-slack-for-4-hosts},
and~\ref{fig.mcb-runtime-vs-tasks-for-4-hosts}, but show results for large
problem instances. The message is the same here: MCB8 is the best algorithm, or
closer on average to the best than the other algorithms. This is clearly seen in
Table~\ref{tab.mcb-degredation-for-64-hosts}, which is similar to
Table~\ref{tab.mcb-degredation-for-4-hosts}, and shows the average and maximum
percent degradation from best for all algorithms for large problem instances.
According to both metrics MCB8 is the best algorithm, with MCB5 performing well
but not as well as MCB8.

In terms of run times, Figure~\ref{fig.mcb-runtime-vs-tasks-for-64-hosts} shows
run times under one-half second for 500 tasks for all of the MCB algorithms.
MCB8 is again the fastest by a tiny margin.

Based on our results we conclude that MCB8 is the best option among the 8
multi-capacity bin packing options.  In all that follows, to avoid graph
clutter, we exclude the 7 other algorithms from our overall results.

\subsubsection{Small Problems}

\begin{figure}
  \centering
  \includegraphics[width=0.45\textwidth]{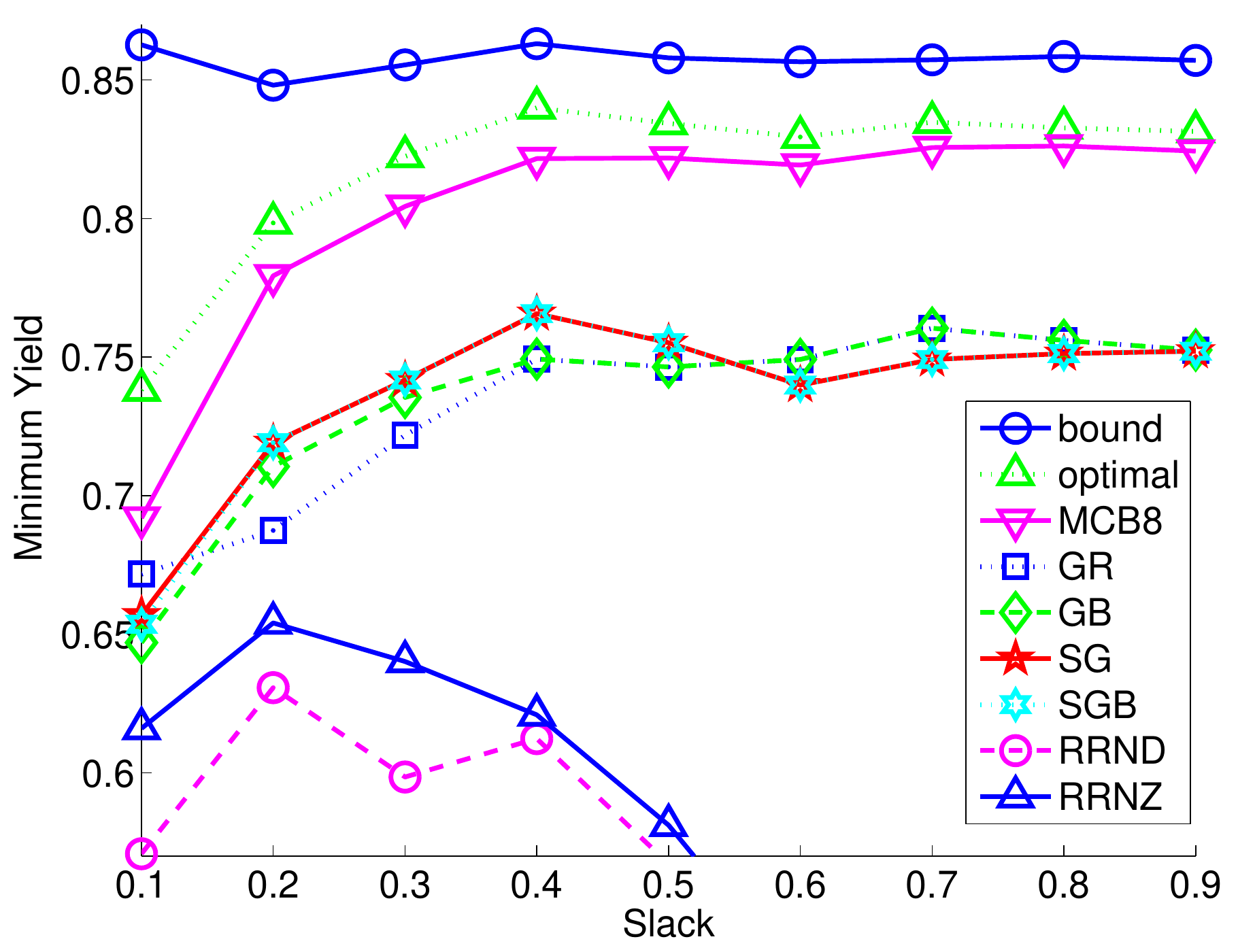}
  \caption{Minimum Yield vs. Slack for small problem instances.}
  \label{fig.overall-minyield-vs-slack-for-4-hosts}
\end{figure}

\begin{figure}
  \centering
  \includegraphics[width=0.45\textwidth]{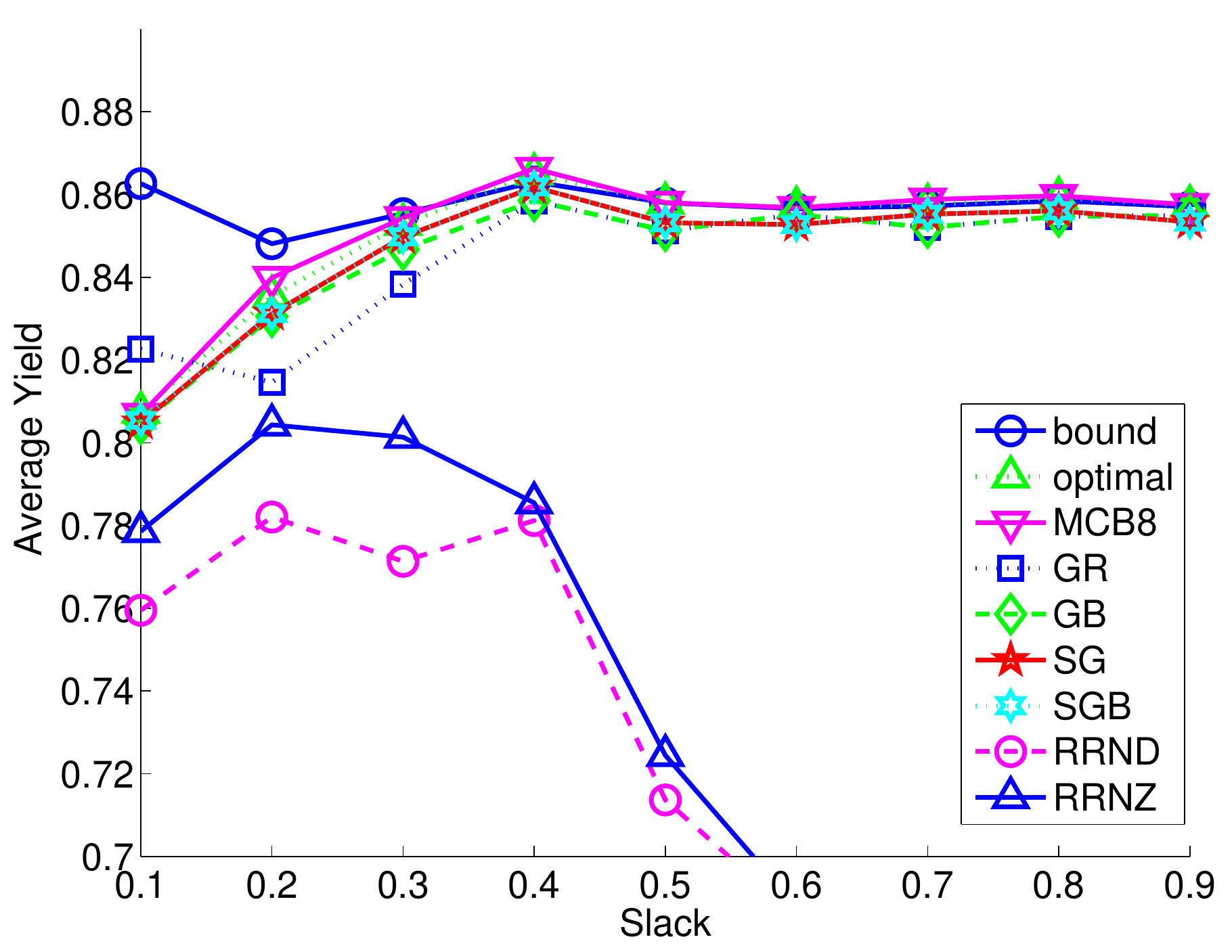}
  \caption{Average Yield vs. Slack for small problem instances.}
  \label{fig.overall-avgyield-vs-slack-for-4-hosts}
\end{figure}

\begin{figure}
  \centering
  \includegraphics[width=0.45\textwidth]{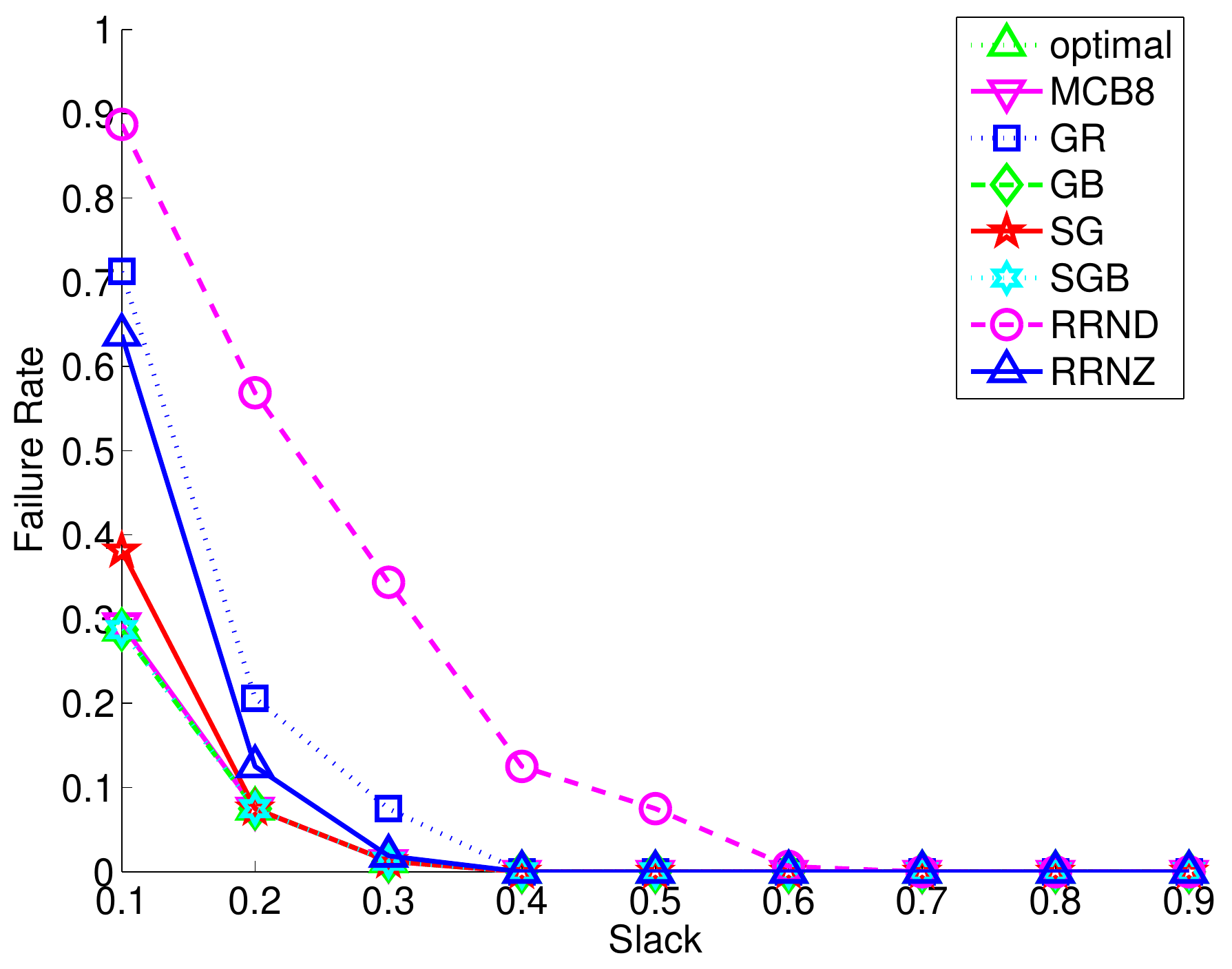}
  \caption{Failure Rate vs. Slack for small problem instances.}
  \label{fig.overall-failrate-vs-slack-for-4-hosts}
\end{figure}

Figure~\ref{fig.overall-minyield-vs-slack-for-4-hosts} shows the achieved
maximum minimum yield versus the memory slack in the system for our algorithms,
the MILP solution, and for the solution of the rational LP, which is an upper
bound on the achievable solution. The solution of the LP is only about 4\%
higher on average than that of the MILP, although it is significantly higher for
very low slack values. The solution of the LP will be interesting for large
problem instances, for which we cannot compute an exact solution.  On average,
the exact MILP solution is about 2\% better than MCB8, and about 11\% to 13\%
better than the greedy algorithms.  All greedy algorithms exhibit roughly the
same performance. The RRND and RRNZ algorithms lead to results markedly
poorer than the other algorithms, with expectedly the RRNZ algorithm slightly 
outperforming the RRND algorithm. Interestingly, once the slack reaches 0.2 the
results of both the RRND and RRNZ algorithms begin to worsen.

Figure~\ref{fig.overall-avgyield-vs-slack-for-4-hosts} is similar to
Figure~\ref{fig.overall-minyield-vs-slack-for-4-hosts} but plots the average
yield. The solution to the rational LP, the MILP solution, the MCB8 solution,
and the solutions produced by the greedy algorithms are all within a few percent
of each other. As in Figure~\ref{fig.overall-minyield-vs-slack-for-4-hosts},
when the slack is lower than 0.2 the relaxed solution is significantly better.

Figure~\ref{fig.overall-failrate-vs-slack-for-4-hosts} plots the failure rates
of our algorithms. The RRND algorithm has the worst failure rate, followed by GR
and then RRNZ. There were a total of 60 instances out of the 1,440 generated 
which were judged to be infeasible because the GLPK solver could not find a
solution for them. We see that the MCB8, SG, and SGB algorithms have failure 
rates that are not significantly larger than that of the exact MILP solution.
Out of the 1,380 feasible instances, the GB and SGB never fail to find a
solution, the MCB8 algorithm fails once, and the SG algorithm fails 15 times.

\begin{figure}
  \centering
  \includegraphics[width=0.45\textwidth]{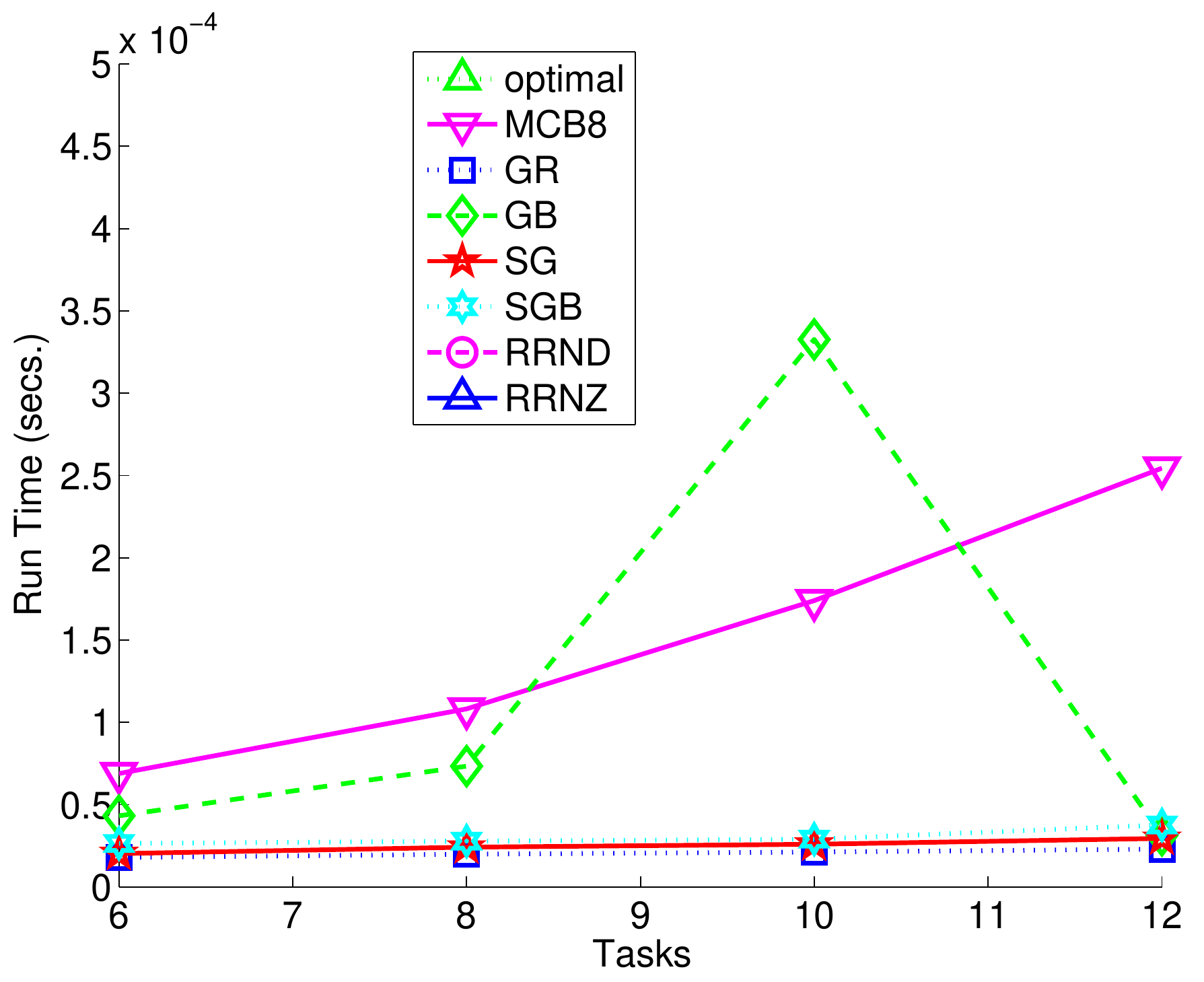}
  \caption{Run time vs. Number of Tasks for small problems instances.}
  \label{fig.overall-runtime-vs-tasks-for-4-hosts}
\end{figure}

Figure~\ref{fig.overall-runtime-vs-tasks-for-4-hosts} shows the run times of the
various algorithms on a 3.2GHz Intel Xeon processor. The computation time of the
exact MILP solution is so much greater than that of the other algorithms that it
cannot be seen on the graph.  Computing the exact solution to the MILP took an
average of 28.7 seconds, however there were 9 problem instances with solutions
that took over 500 seconds to compute, and a single problem instance that
required 11,549.29 seconds (a little over 3 hours) to solve. For the small
problem instances the average run times of all greedy algorithms and of the MCB8
algorithm are under 0.15 milliseconds, with the simple GR and SG algorithms
being the fastest. The RRND and RRNZ algorithms are significantly slower, with
run times a little over 2 milliseconds on average; they also cannot be seen on
the graph.

\subsubsection{Large Problems}
~\\

Figures~\ref{fig.overall-minyield-vs-slack-for-64-hosts},
\ref{fig.overall-avgyield-vs-slack-for-64-hosts},
\ref{fig.overall-failrate-vs-slack-for-64-hosts}, and
\ref{fig.overall-runtime-vs-tasks-for-64-hosts} are similar to
Figures~\ref{fig.overall-minyield-vs-slack-for-4-hosts},
\ref{fig.overall-avgyield-vs-slack-for-4-hosts},
\ref{fig.overall-failrate-vs-slack-for-4-hosts}, and
\ref{fig.overall-runtime-vs-tasks-for-4-hosts} respectively, but for
large problem instances.  In
Figure~\ref{fig.overall-minyield-vs-slack-for-64-hosts} we can see that MCB8
algorithm achieves far better results than any other heuristic. Furthermore,
MCB8 is extremely close to the upper bound as soon as the slack is 0.3 or larger
and is only 8\% away from this upper bound when the slack is 0.2. When the slack
is 0.1, MCB8 is 37\% away from the upper bound but we have seen with the small
problem instances that in this case the upper bound is significantly larger than
the actual optimal (see Figure~\ref{fig.overall-minyield-vs-slack-for-4-hosts}).

The performance of the greedy algorithms has worsened relative to the rational
LP solution, on average 20\% lower for slack values larger than 0.2. The GR and
GB algorithms perform nearly identically, showing that backtracking does not
help on the large problem instances. The RRNZ algorithm is again a poor
performer, with a profile that, unexpectedly, drops as slack increases. The
RRND algorithm not only achieved the lowest values for minimum yield, but also
completely failed to solve any instances of the problem for slack values less
than 0.4.

Figure~\ref{fig.overall-avgyield-vs-slack-for-64-hosts} shows the achieved
average yield values. The MCB8 algorithm again tracks the optimal for slack
values larger than 0.3. A surprising observation at first glance is that the
greedy algorithms manage to achieve higher average yields than the optimal or
MCB algorithms. This is due to their lower achieved minimum yields. Indeed, with
a lower minimum yield, average yield maximization is less constrained, making it
possible to achieve higher average yield than when starting from and allocation
optimal for the minimum yield. The greedy algorithms thus trade off fairness for
higher average performance. The RRNZ algorithm starts out doing well for average
slack, even better than GR or GB when the slack is low, but does much worse as
slack increases.

Figure~\ref{fig.overall-failrate-vs-slack-for-64-hosts} shows that for large
problem instances the GB and SGB algorithms have nearly as many failures as the
GR and SG algorithms when slack is low. This suggests that the arbitrary bound
of 500,000 placement attempts when backtracking, which was more than sufficient
for the small problem set, has little affect on overall performance for the
large problem set. It could thus be advisable to set the bound on the number of
placement attempts based on the size of the problem set and time allowed for
computation. The RRND algorithm is the only algorithm with a significant number
of failures for slack values larger than 0.3. The SG, SGB and MCB8 algorithms
exhibit the lowest failure rates, about 40\% lower than that experienced by the
other greedy and RRNZ algorithms, and more than 14 times lower than the failure
rate of the RRND algorithm. Keep in mind that, based on our experience with the
small problem set, some of the problem instances with small slacks may not be
feasible at all. 

Figure~\ref{fig.overall-runtime-vs-tasks-for-64-hosts} plots the average time
needed to compute the solution to \probzero on a 3.2GHz Intel Xeon for all the
algorithms versus the number of jobs. The RRND and RRNZ algorithms require
significant time, up to roughly 650 seconds on average for 500 tasks, and so 
cannot be seen at the given scale. This is attributed to solving the relaxed
MILP using GLPK. Note that this time could be reduced significantly by using a 
faster solver (e.g., CPLEX~\cite{cplex}). The GB and SGB algorithms require
significantly more time when the number of tasks is small. This is because the
failure rate decreases as the number of tasks increases. For a given set of
parameters, increasing the number of tasks decreases granularity. Since there is
a relatively large number of unsolvable problems when the number of tasks is
small, these algorithms spend a lot of time backtracking and searching though
the solution space fruitlessly, ultimately stopping only when the bounded number
of backtracking attempts is reached. The greedy algorithms are faster than the
MCB8 algorithm, returning solutions in 15 to 20 milliseconds on average for 500
tasks as compared to nearly half a second for MCB8. Nevertheless, less than .5
seconds for 500 tasks is clearly acceptable in practice.

\begin{figure}
  \centering
  \includegraphics[width=0.45\textwidth]{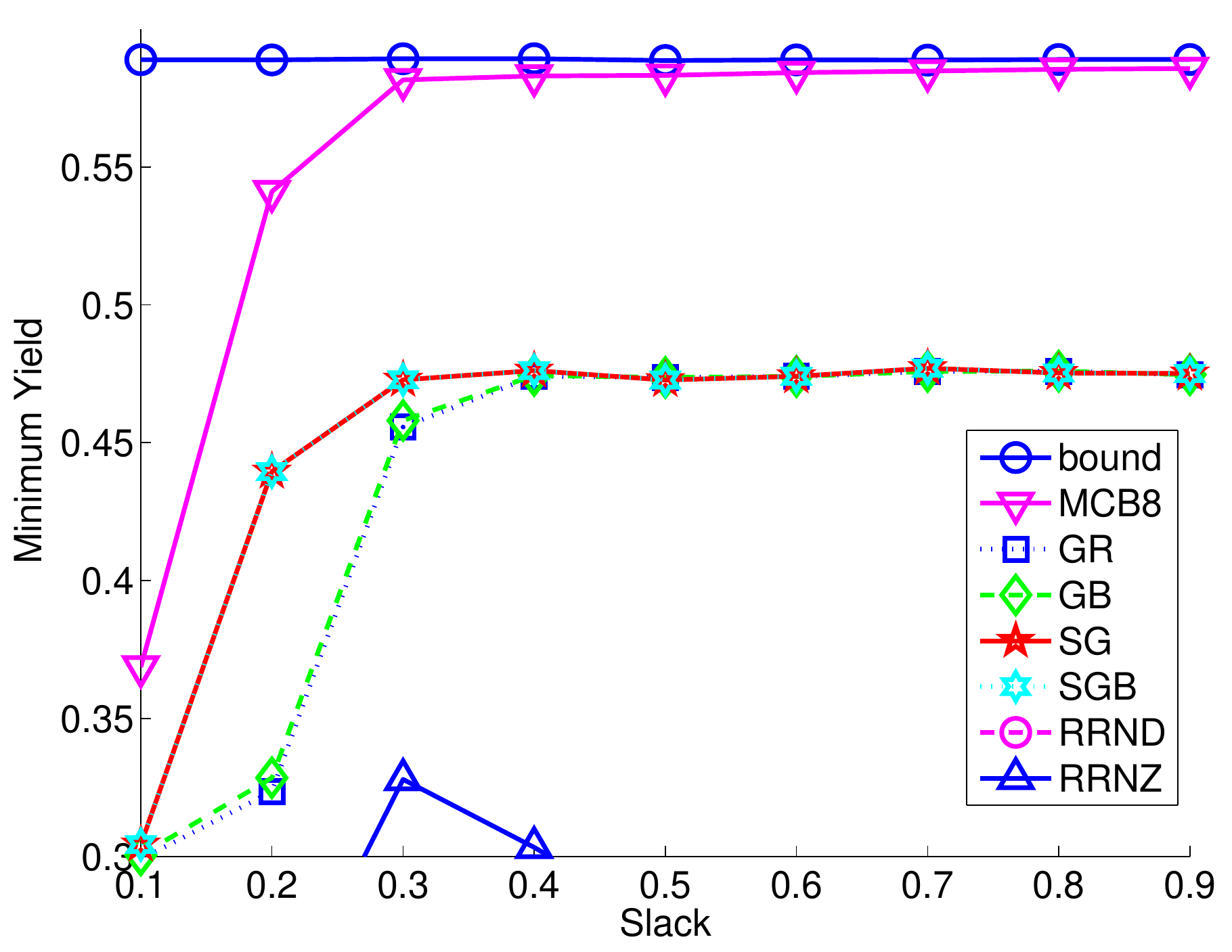}
  \caption{Minimum Yield vs. Slack for large problem instances.}
  \label{fig.overall-minyield-vs-slack-for-64-hosts}
\end{figure}

\begin{figure}
  \centering
  \includegraphics[width=0.45\textwidth]{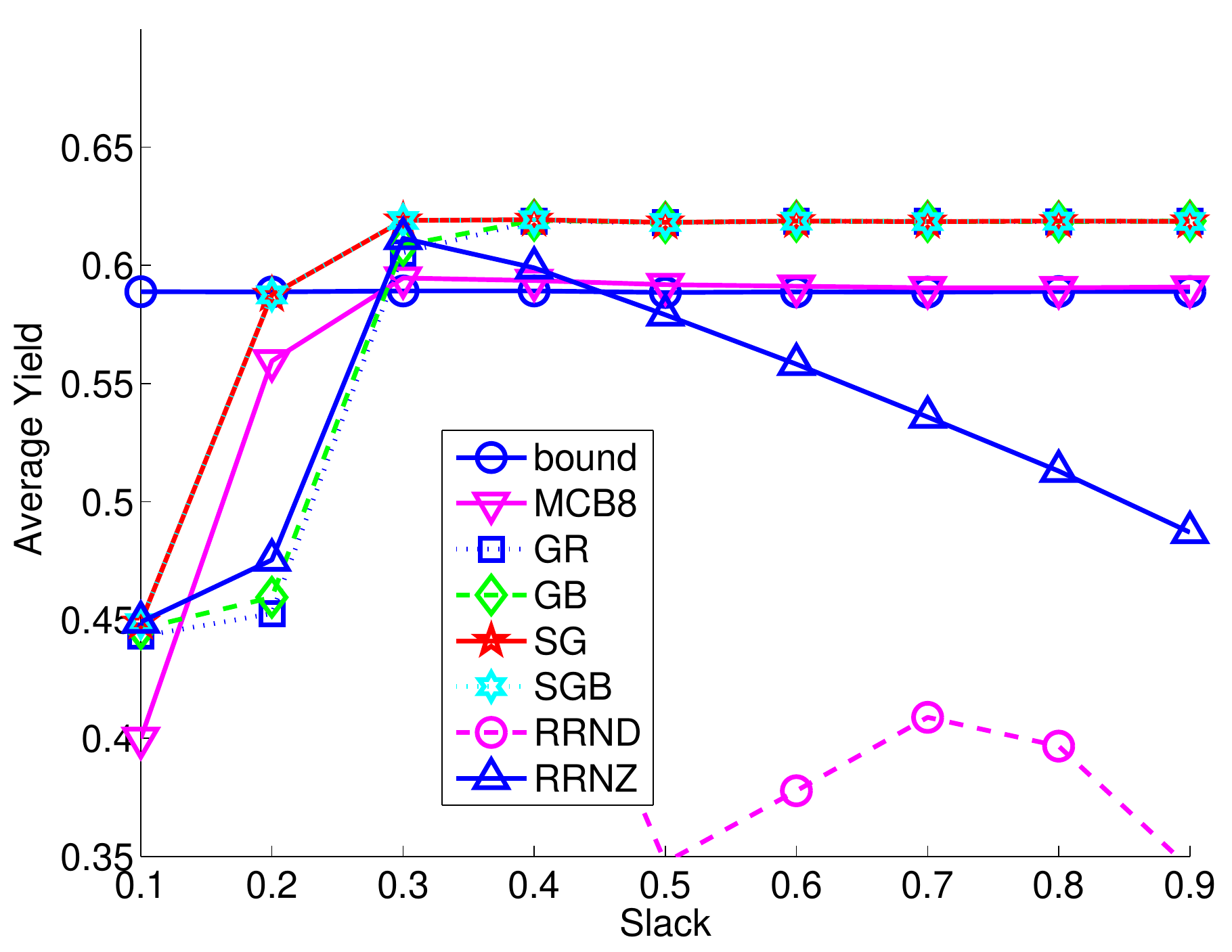}
  \caption{Average Yield vs. Slack for large problem instances.}
  \label{fig.overall-avgyield-vs-slack-for-64-hosts}
\end{figure}

\begin{figure}
  \centering
  \includegraphics[width=0.45\textwidth]{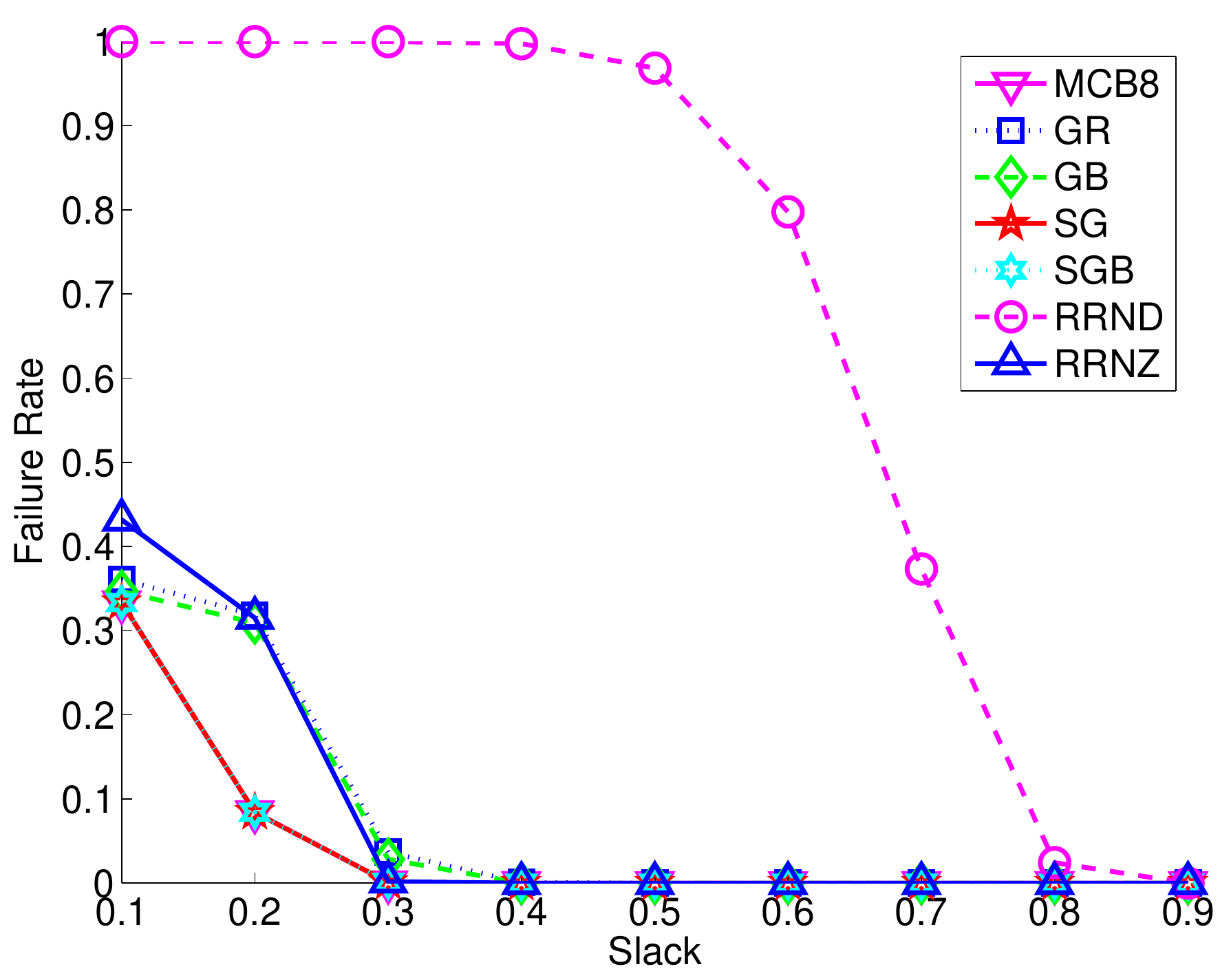}
  \caption{Failure Rate vs. Slack for large problem instances.}
  \label{fig.overall-failrate-vs-slack-for-64-hosts}
\end{figure}

\begin{figure}
  \centering
  \includegraphics[width=0.45\textwidth]{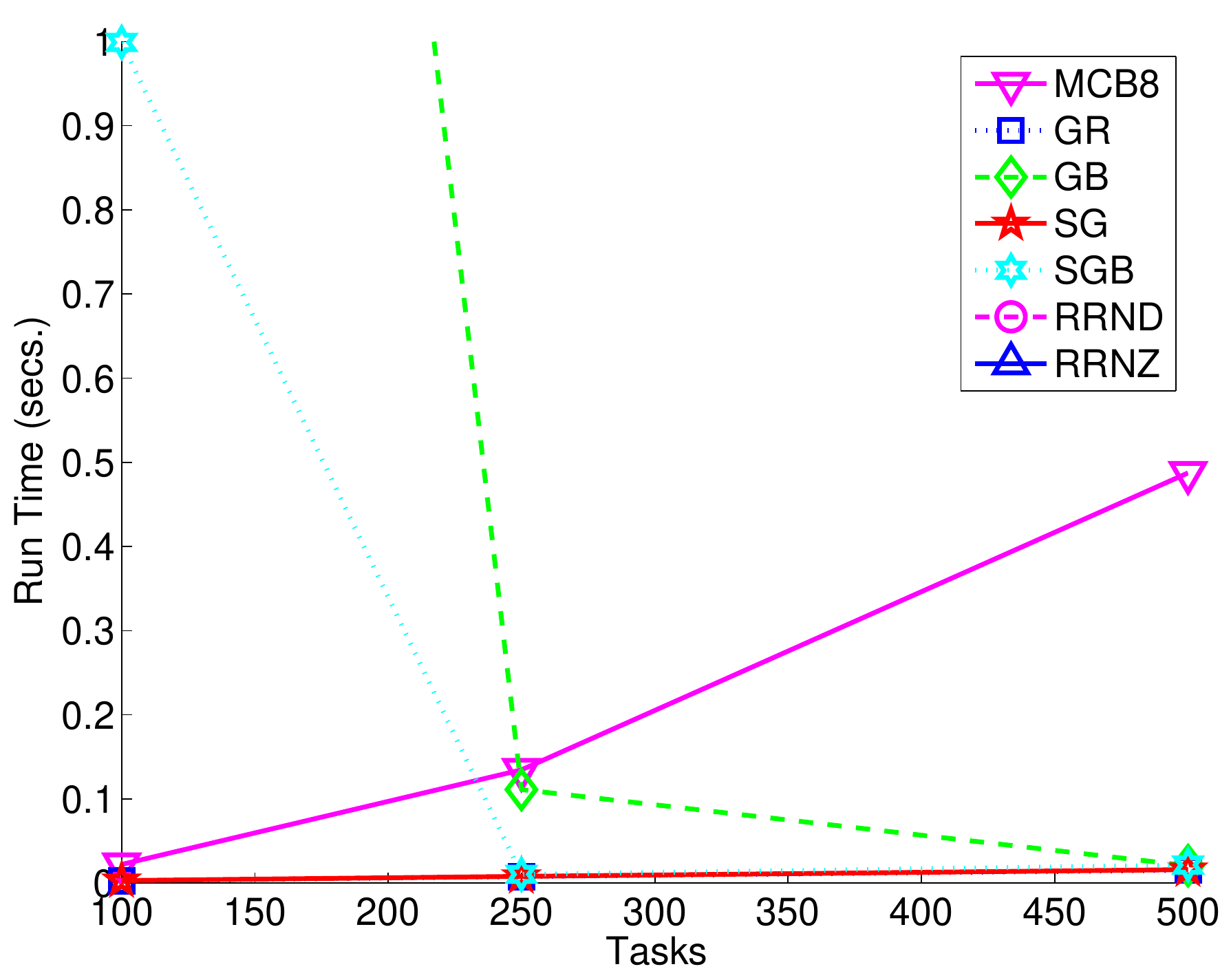}
  \caption{Run time vs. Number of Tasks for large problem instances.}
  \label{fig.overall-runtime-vs-tasks-for-64-hosts}
\end{figure}

\subsubsection{Discussion}

Our main result is that the multi-capacity bin packing algorithm that sorts
tasks in descending order by their largest resource requirement (MCB8) is the
algorithm of choice. It outperforms or equals all other algorithm nearly across
the board in terms of minimum yield, average yield, and failure rate, while
exhibiting relatively low run times. The sorted greedy algorithms (SG or SGB)
lead to reasonable results and could be used for very large numbers of tasks,
for which the run time of MCB8 may become too high. The use of backtracking in
the algorithms GB and SGB led to performance improvements for small problem sets
but not for large problem sets, suggesting that some sort of backtracking system
with a problem-size- or run-time-dependent bound on the number of branches to
explore could potentially be effective.

\section{Parallel Jobs}
\label{sec.parallel}

\subsection{Problem Formulation}

In this section we explain how our approach and algorithms can be easily
extended to handle parallel jobs that consist of multiple tasks (relaxing
assumption H3). We have thus far only concerned ourselves with independent jobs
that are both indivisible and small enough to run on a single machine. However,
in many cases users may want to split up jobs into multiple tasks, either
because they wish to use more CPU power in order to return results more quickly
or because they wish to process an amount of data that does not fit comfortably
within the memory of a single machine.

One na\"ive way to extend our approach to parallel jobs would be to simply
consider the tasks of a job independently. In this case individual tasks of the
same job could then receive different CPU allocations. However, in the vast
majority of parallel jobs it is not useful to have some tasks run faster than 
others as either the job makes progress at the rate of the slowest task or the
job is deemed complete only when all tasks have completed. Therefore, we opt to
add constraints to our linear program to enforce that the CPU allocations of
tasks within the same job must be identical. It would be straightforward to have
more sophisticated constraints if specific knowledge about a particular job is
available (e.g., task A should receive twice as much CPU as task B).

Another important issue here is the possibility of gaming the system when
optimizing the average yield. When optimizing the minimum yield, a division of a
job into multiple tasks that leads to a higher minimum yield benefits all jobs.
However, when considering the average yield optimization, which is done in our
approach as a second round of optimization, a problem arises because the average
yield metric favors small tasks, that is, tasks that have low CPU requirements.
Indeed, when given the choice to increase the CPU allocation of a small task or
of a larger task, for the same additional fraction of CPU, the absolute yield
increase would be larger for the small task, and thus would lead to a higher
average yield. Therefore, an unscrupulous user might opt for breaking his/her
job into unnecessarily many smaller tasks, perhaps hurting the parallel
efficiency of the job, but acquiring an overall larger portion of the total
available CPU resources, which could lead to shorter job execution time. To
remedy this problem we use a per-job yield metric (i.e., total CPU allocation
divided by total CPU requirements) during the average yield optimization phase.

The linear programming formulation with these additional considerations and
constraints is very similar to that derived in Section~\ref{sec.milp}.  We again
consider jobs $1..J$ and hosts $1..H$. But now each job $i$ consists of $T_i$
tasks. Since these jobs are constrained to be uniform, $\alpha_i$ represents the
maximum CPU consumption and $m_i$ represents the maximum memory consumption of
all tasks $k$ of job $i$. The integer variables $e_{ikj}$ are constrained to be
either $0$ or $1$ and represent the absence or presence of task $k$ of job $i$
on host $j$. The variables $\alpha_{ikj}$ represent the amount of CPU allocated
to task $k$ of job $i$ on host $j$.

\begin{eqnarray}
\label{eq.parallel}
\forall i,k,j  & \quad e_{ikj} \in \N,\\
\forall i,k,j  & \quad \alpha_{ikj} \in \Q,\\
\forall i,k,j  & 0 \leq e_{ikj} \leq 1,\\
\forall i,k,j  & 0 \leq \alpha_{ikj} \leq e_{ikj},\\
\forall i,k    & \sum_{j=1}^{H} e_{ikj} = 1,\\
\forall j      & \sum_{i=1}^{J} \sum_{k=1}^{T_i} \alpha_{ikj} \leq 1,\\
\forall j      & \sum_{i=1}^{J} \sum_{k=1}^{T_i} e_{ikj} m_{i} \leq 1,\\
\forall i,k    & \sum_{j=1}^{H} \alpha_{ikj} \leq \alpha_i,\\
\forall i,k,k' & \sum_{j=1}^{H} \alpha_{ikj} = \sum_{j=1}^{H} \alpha_{ik'j},\\
\forall i      & \sum_{j=1}^{H} \sum_{k=1}^{T_i} 
                \frac{\alpha_{ikj}}{T_i \times \alpha_{i}} \geq Y
\end{eqnarray}

Note that the final constraint is logically equivalent to the per-task yield
since all tasks are constrained to have the same CPU allocation. The reason for
writing it this way is to highlight that in the second phase of optimization one
should maximize the average per-job yield rather than the average per-task
yield.

\subsection{Results}

The algorithms described in Section~\ref{sec.algs} for the case of sequential
jobs can be used directly for minimum yield maximization for parallel jobs. The
only major difference is that the average per-task yield optimization phase
needs to be changed for an average per-job optimization phase.  As with the
per-task optimization, we make the simplifying assumption that task placement
decisions cannot be changed during this phase of the optimization. This
simplification removes not only the difficulty of solving a MILP, but also
allows us to avoid the enormous number of additional constraints which would be
required to make sure that all of a given job's tasks receive the same
allocation while keeping the problem linear.

We present results only for large problem instances as defined in
Section~\ref{sec.methodology}. We use the same experimental methodology as
defined there as well. We only need a way to decide how many tasks comprise a
parallel job. To this end, we use the parallel workload model proposed
in~\cite{workload_model_03}, which models many characteristics of parallel
workloads (derived based on statistical analysis of real-world batch system
workloads). The model for the number of tasks in a parallel job uses a two-stage
log-uniform distribution biased towards powers of two. We instantiate this model
using the same parameters as in~\cite{workload_model_03}, assuming that jobs can
consist of between 1 and 64 tasks.

\begin{figure}
  \centering
  \includegraphics[width=0.45\textwidth]{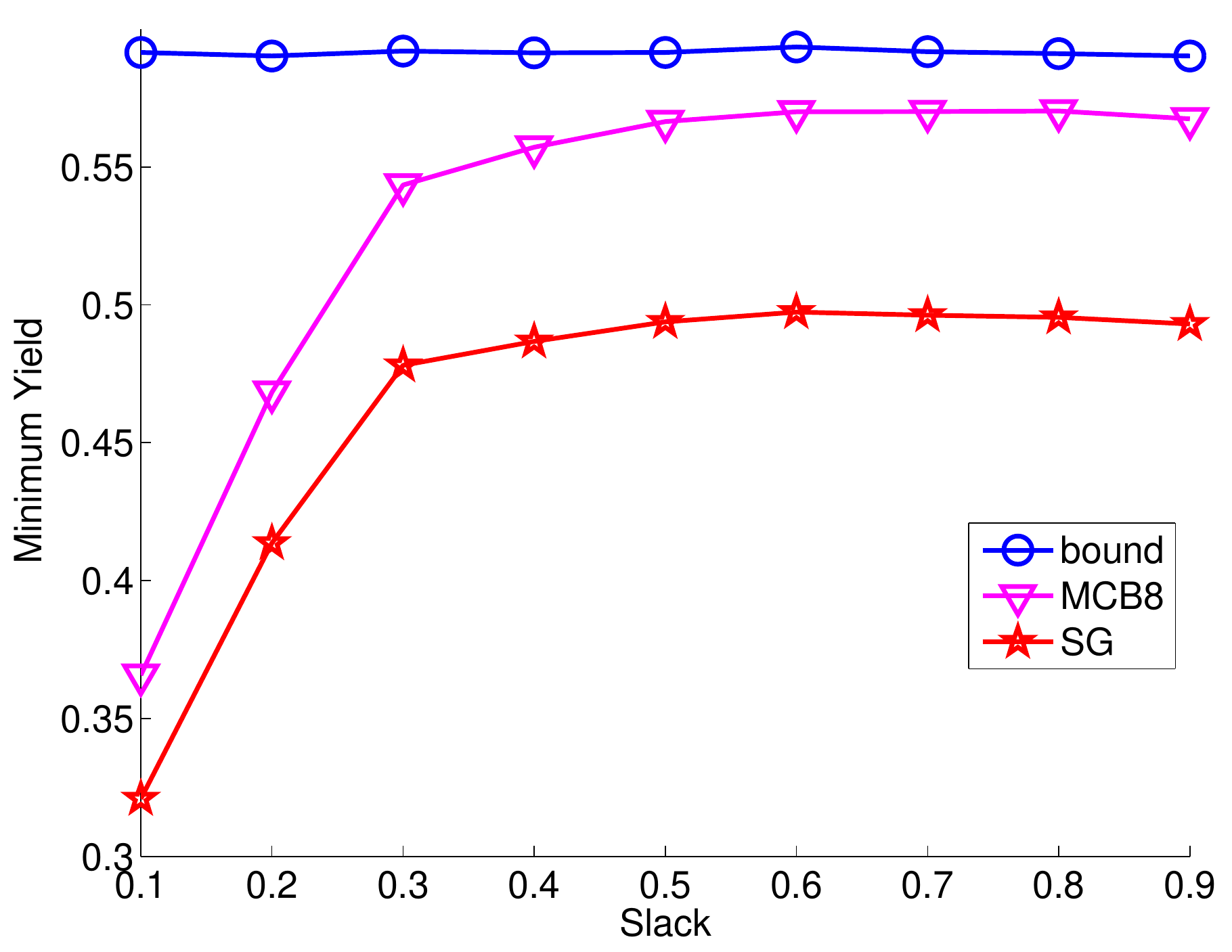}
  \caption{Minimum Yield vs. Slack for large problem instances for parallel
jobs.}
  \label{fig.parallel-minyield-vs-slack-for-64-hosts}
\end{figure}

Figure~\ref{fig.parallel-minyield-vs-slack-for-64-hosts} shows results for the
SG and the MCB8 algorithms. We exclude all other greedy algorithms as they were
all shown to be outperformed by SG, all other MCB algorithms because they were
all shown to be outperformed by MCB8, as well as the RRND and RRNZ algorithms
which were shown to perform poorly. The figure also shows the upper bound on
optimal obtained assuming that $e_{ij}$ variables can take rational values. We
see that MCB8 outperforms the SGB algorithm significantly and is close to the
upper bound on optimal for slacks larger than $0.3$. 

\begin{figure}
  \centering
  \includegraphics[width=0.45\textwidth]{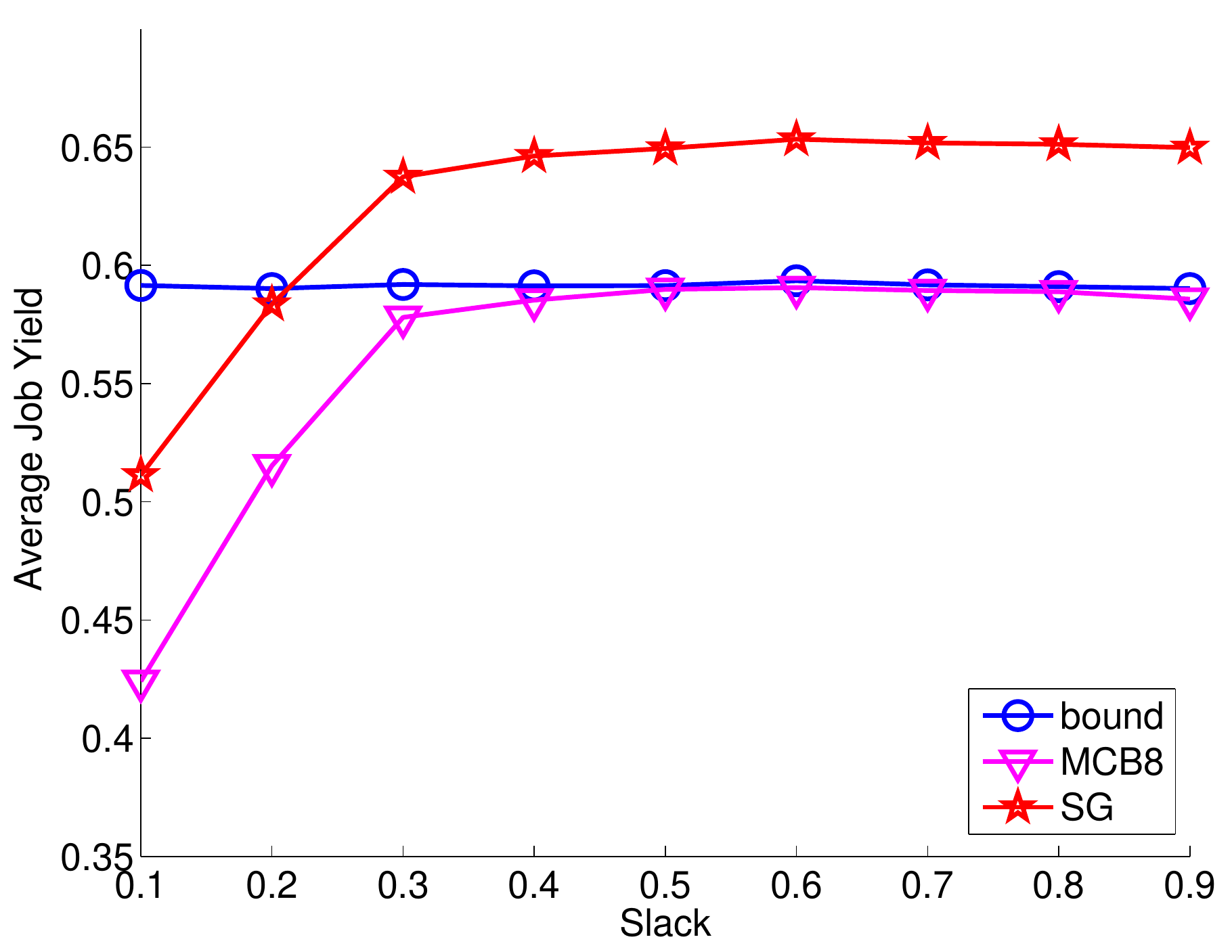}
  \caption{Average Yield vs. Slack for large problem instances for parallel
jobs.}
  \label{fig.parallel-avgjobyield-vs-slack-for-64-hosts}
\end{figure}

Figure~\ref{fig.parallel-avgjobyield-vs-slack-for-64-hosts} shows the average
job yield. We see the same phenomenon as in
Figure~\ref{fig.overall-avgyield-vs-slack-for-64-hosts}, namely that the greedy
algorithm can achieve higher average yield because it starts from a lower
minimum yield, and thus has more options to push the average yield higher 
(thereby improving average performance at the expense of fairness). 

\begin{figure}
  \centering
  \includegraphics[width=0.45\textwidth]{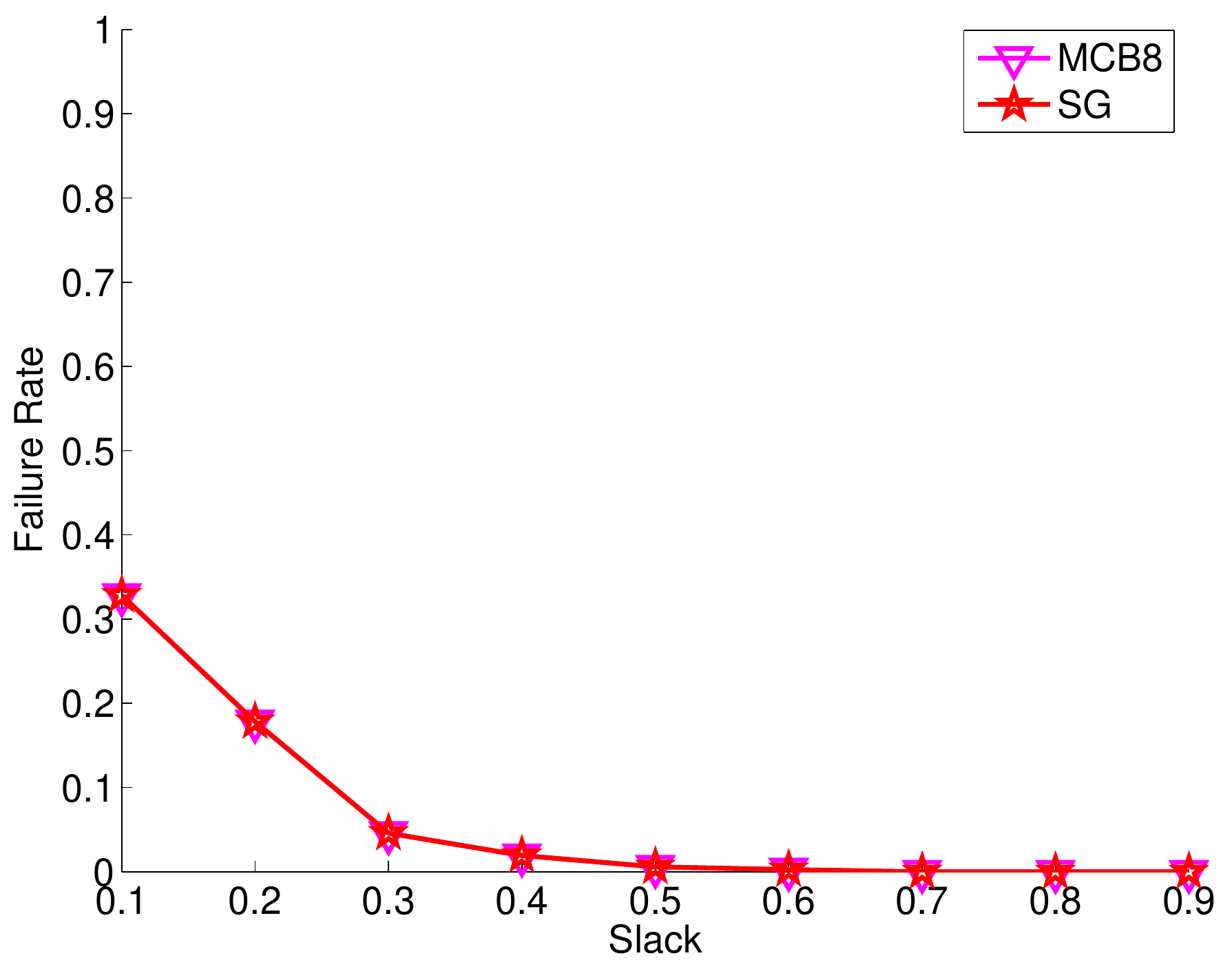}
  \caption{Failure Rate vs. Slack for large problem instances for parallel
jobs.}
  \label{fig.parallel-failrate-vs-slack-for-64-hosts}
\end{figure}

\begin{figure}
  \centering
  \includegraphics[width=0.45\textwidth]{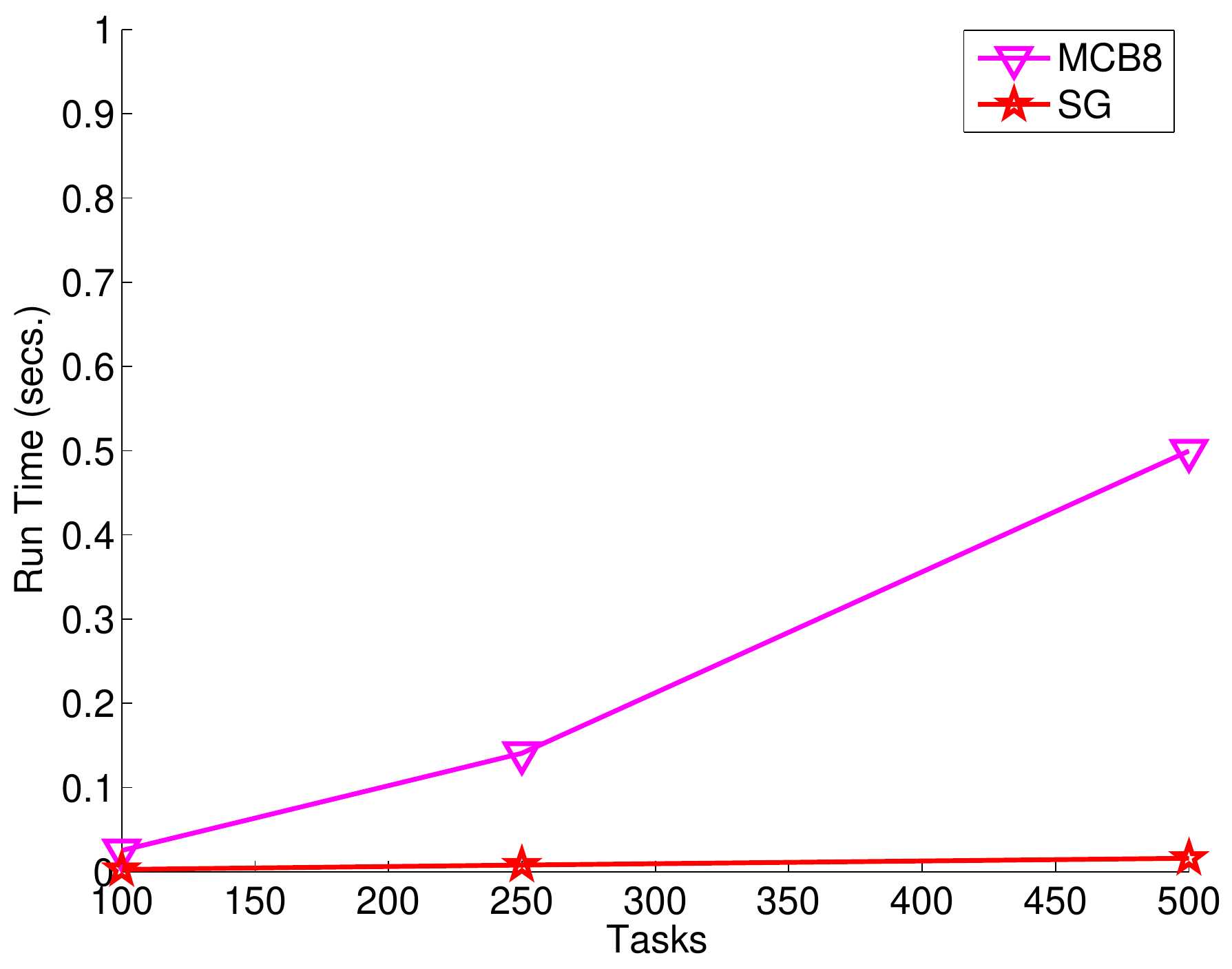}
  \caption{Runtime vs. Number of Tasks for large problem instances for parallel
jobs.}
  \label{fig.parallel-runtime-vs-tasks-for-64-hosts}
\end{figure}

Figure~\ref{fig.parallel-failrate-vs-slack-for-64-hosts} shows the failure rates
of the MCB8 and SG algorithms, which are identical. Finally
Figure~\ref{fig.parallel-runtime-vs-tasks-for-64-hosts} shows the run time of
both algorithms. We see that the SG algorithm is much faster than the MCB8
algorithm (by roughly a factor 32 for 500 tasks). Nevertheless, MCB8 can still
compute an allocation in under one half a second for 500 tasks.

Our conclusions are similar to the ones we made when examining results for
sequential jobs: in the case of parallel jobs the BCB8 algorithm is the
algorithm of choice for optimizing minimum yield, while the SGB algorithm could
be an alternate choice if the number of tasks is very large.  

\section{Dynamic Workloads}
\label{sec.dynamic}

In this section we study resource allocation in the case when assumption H4 no
longer holds, meaning that the workload is no longer static. We assume that job
resource requirements can change and that jobs can join and leave the system.
When the workload changes, one may wish to adapt the schedule to reach a new
(nearly) optimal allocation of resources to the jobs. This adaptation can entail
two types of actions: (i)~modifying the CPU fractions allocated to some jobs;
and (ii)~migrating jobs to different physical hosts. In what follows we extend
the linear program formulation derived in Section~\ref{sec.milp} to account for
resource allocation adaptation. We then discuss how current technology can be
used to implement adaptation with virtual clusters.

\subsection{Mixed-Integer Linear Program Formulation}

One difficult question for resource allocation adaptation, regardless of the
context, is whether the adaptation is ``worth it.'' Indeed, adaptation often
comes with an overhead, and this overhead may lead to a loss of performance. In
the case of virtual cluster scheduling, the overhead is due to VM migrations.
The question of whether adaptation is worthwhile is often based on a time
horizon (e.g., adaptation is not worthwhile if the workload is expected to
change significantly in the next 5 minutes)~\cite{adaptation_1,autopilot}. In
virtual cluster scheduling, as defined in this paper, jobs do not have time 
horizons. Therefore, in principle, the scheduler cannot reason about when
resource needs will change. It may be possible for the scheduler to keep track
of past workload behavior to forecast future workload behavior. Statistical
workload models have been built (see~\cite{workload_model_03,
li_supercomputing_2008} for models and literature reviews). Techniques to make
predictions based on historical information have been developed
(see~\cite{tsafrir_tpds2007} for task execution time models and a good
literature review). Making sound short-term decisions for resource allocation
adaptation requires highly accurate predictions, so as to carry out precise
cost-benefit analyses of various adaptation paths. Unfortunately, accurate point
predictions (rather than statistical characterizations) are elusive due to the
inherently statistical and transient nature of the workload, as seen in the
aforementioned works. Furthermore, most results in this area are obtained for
batch scheduling environments with parallel scientific applications, and it is
not clear whether the obtained models would be applicable in more general
settings (e.g., cloud computing environments hosting internet services).

Faced with the above challenge, rather than attempting arduous statistical
forecasting of adaption cost and pay-off, we side-step the issue and propose a
pragmatic approach. We consider schedule adaptation that attempts to achieve the
best possible yield, but so that job migrations do not entail moving more than
some fixed number of bytes, $B$ (e.g., to limit the amount of network load due
to schedule adaptation). If $B$ is set to $0$, then the adaptation will do the
best it can without using migration whatsoever. If $B$ is above the sum of the
job sizes (in bytes of memory requirement), then all jobs could be migrated.

It turns out that this adaptation scheme can be easily formulated as a
mixed-integer linear program. More generally, the value of $B$ can be chosen so
that it achieves a reasonable trade-off between overhead and workload
dynamicity. Choosing the best value for $B$ for a particular system could
however be difficult and may need to be adaptive as most workloads are
non-stationary. A good approach is likely to pick relatively smaller values of
$B$ for more dynamic workload. We leave a study of how to best tune parameter
$B$ for future work.

We use the same notations and definitions as in Section~\ref{sec.milp}. In
addition, we consider that some jobs are already assigned to a host:
$\bar{e}_{ij}$ is equal to $1$ if job $i$ is already running on host $j$, and
$0$ otherwise. For reasons that will be clear after we explain our constraints,
we simply set $\bar{e}_{ij}$ to $1$ for all $j$ if job $i$ corresponds to a
newly arrived job. Newly departed jobs need not be taken into account. We can
now write a new set of constraints as follows:

\begin{eqnarray}
\label{eq.constraints_adaptation}
\forall i,j & \quad e_{ij} \in \N,\\
\forall i,j & \quad \alpha_{ij} \in \Q,\\
\forall i,j & 0 \leq e_{ij} \leq 1,\\
\forall i,j & 0 \leq \alpha_{ij} \leq e_{ij},\\
\forall i   & \sum_{j=1}^{H} e_{ij} = 1,\\
\forall j   & \sum_{i=1}^{T} \alpha_{ij} \leq 1,\\
\forall j   & \sum_{i=1}^{T} e_{ij} m_{i} \leq 1 \;\\
\forall i   & \sum_{j=1}^{H} \alpha_{ij} \leq \alpha_i \;\\
\forall i   & \sum_{j=1}^{H}\frac{\alpha_{ij}}{\alpha_{i}} \geq Y\\
            & \sum_{i=1}^{T} \sum_{j=1}^{H} (1 - \bar{e}_{ij})e_{ij} m_i \leq B
\end{eqnarray}

The objective, as in Section~\ref{sec.milp}, is to maximize $Y$. The only new
constraint is the last one. This constraint simply states that if job  $i$ is
assigned to a host that is different from the host to which it was assigned
previously, then it needs to be migrated. Therefore, $m_i$ bytes need to be
transferred. These bytes are summed over all jobs in the system to ensure that
the total number of bytes communicated for migration purposes does not exceed
$B$. Note that this is still a linear program as $\bar{e}_{ij}$ is not a
variable but a constant. Since for newly arrived jobs we set all $\bar{e}_{ij}$
values to $1$, we can see that they do not contribute to the migration cost.
Note that removing $m_i$ in the last constraint would simply mean that $B$ is a
bound on the total number of job migrations allowed during schedule adaptation.

We leave the development of heuristic algorithms for solving the above linear
program for future work. 

\subsection{Technology Issues for Resource Allocation Adaptation}

In the linear program in the previous section nowhere do we account for the time
it takes to migrate a job. While a job is being migrated it is presumably
non-responsive, which impacts the yield.  However, modern VM monitors support
``live migration'' of VM instances, which allows migrations with only
milliseconds of unresponsiveness~\cite{xen-migrate-nsdi05}. There could be a
performance degradation due to memory pages being migrated between two physical
hosts. Resource allocation adaptation also requires quick modification of the
CPU share allocated to a VM instance (assumption H5). We validate this
assumption in Section~\ref{sec.vmtech} and find that, indeed, CPU shares can be
modified accurately in under a second.

\section{Evaluation of the Xen Hypervisor}
\label{sec.vmtech}

Assumption H5 in Section~\ref{sec.assumptions} states that VM technology allows
for precise, low-overhead, and quickly adaptable sharing of the computational
capabilities of a host across CPU-bound VM instances.  Although this seems like
a natural expectation, we nevertheless validate this assumption with
state-of-the-art virtualization technology, namely, the Xen VM
monitor~\cite{barham03xav}. While virtualization can happen inside the operating
system (e.g, Virtual PC~\cite{Virtual_PC}, VMWare~\cite{VMWare}), Xen runs
between the hardware and the operating system. It thus requires either a
modified operating system (``paravirtualization'') or hardware support for
virtualization (``hardware virtualization''~\cite{HW_Virt}). In this work we use
Xen 3.1 on a dual-CPU 64-bit machine with paravirtualization. All our VM
instances use identical 64-bit Fedora images, are allocated 700MB of RAM, and
run on the same physical CPU. The other CPU is used to run the experiment
controller. All our VM instances perform continuous CPU-bound computations, that
is, $100\times 100$ double precision matrix multiplications using the LAPACK
DGEMM routine~\cite{LAPACK}.  

Our experiments consist in running from one to ten VM instances with specified
``cap values'', which Xen uses to control what fraction of the CPU is allocated
to each VM. We measure the effective compute rate of each VM instance (in number
of matrix multiplications per seconds). We compare this rate to the expected
rate, that is, the cap value times the compute rate measured on the raw
hardware. We can thus ascertain both the accuracy and the overhead of the
CPU-sharing in Xen. We also conduct experiments in which we change cap values
on-the-fly and measure the delay before the effective compute rates are in
agreement with the new cap values.

Due to space limitations we only provide highlights of our results and refer the
reader to a technical report for full details~\cite{ics_2008_05-01}. We found
that Xen imposes a minimal overhead (on average a 0.27\% slowdown). We also
found that the absolute error between the effective compute rate and the
expected compute rate was at most 5.99\% and on average 0.72\%. In terms of
responsiveness, we found that the effective compute rate of a VM becomes
congruent with a cap value less than one second after that cap value was
changed. We conclude that, in the case of CPU-bound VM instances, CPU-sharing in
Xen is sufficiently accurate and responsive to enable fractional and dynamic 
resource allocations as defined in this paper.  

\section{Related Work}
\label{sec.related}

The use of virtual machine technology to improve parallel job reliability,
cluster utilization, and power efficiency is a hot topic for research and
development, with groups at several universities and in industry actively
developing resource management systems~\cite{Usher, virtual_center,
XenEnterprise, Laura_vtdc2007, Ruth_icac2006}. This paper builds on top of such
research, using the resource manager intelligently to optimize a user-centric
metric that attempts to capture common ideas about fairness among the users of
high-performance systems.

Our work bears some similarities with gang scheduling.  However, traditional
gang scheduling approaches suffer from problems due to memory pressure and the
communication expense of coordinating context switches across multiple
hosts~\cite{gang_scheduling,gang_scheduling_mem}. By explicitly considering task
memory requirements when making scheduling decisions and using virtual machine
technology to multiplex the CPU resources of individual hosts our approach
avoids these problems.

Above all, our approach is novel in that we define and optimize for a
user-centric metric which captures both fairness and performance in the face of
unknown time horizons and fluctuating resource needs. Our approach has the
additional advantage of allowing for interactive job processes.  

\section{Limitations and Future Directions}
\label{sec.future}

In this work we have made two key assumptions. The first assumption is that VM
instances are CPU-bound (assumption H1), which made it possible to validate
assumption H5 in Section~\ref{sec.vmtech}. However, in reality, VM instances may
have composite needs that span multiple resources, including the network, the
disk, and the memory bus. The second assumption is that resource needs are known
(assumption H2). However, this typically does not hold true in practice as users
do not know precise resource needs of their applications. When assumption H1
does not hold, the challenge is to model composite resource needs in the
definition of the resource allocation problem, and to share these various
resources among VM instances in practice.

In practice, CPU and network resources are strongly dependent within a virtual
machine monitor environment. To ensure secure isolation, VM monitors interpose
on network communication, adding CPU overhead as a result. Experience has shown
that, because of this dependence, one can capture network needs in terms of
additional CPU need~\cite{schedulers:per}. Therefore, it should be
straightforward to modify our approach to account for network resource usage. In
terms of disk usage, we note that virtual cluster environments typically use
network-attached storage to simplify VM migration. As a result, disk usage is
subsumed in network usage. In both cases one should then be able to both model
and precisely share network and disk usage. Much more challenging is the
modeling and sharing of the memory bus usage, due to complex and deep memory
hierarchies on multi-core processors. However, current work on Virtual Private
Machines points to effective ways for achieving sharing and performance
isolation among VM instances of microarchitecture
resources~\cite{Nesbit_IEEE2008}, including the memory
hierarchy~\cite{Nesbit_ISCA2007}.

In terms of discovering VM instance resource needs, a first approach is to use
standard services for tracking VM resource usage across a cluster and collecting
the information as input into a cluster system scheduler (e.g., the XenMon VM
monitoring facility in Xen~\cite{xenmon}. Application resource needs inside a
VM instance can be discovered via a combination of introspection and
configuration variation. With introspection, for example, one can deduce
application CPU needs by inferring process activity inside of
VMs~\cite{antfarm-usenix06}, and memory pressure by inferring memory page
eviction activity~\cite{geiger-asplos06}. This kind of monitoring and inference 
provides one set of data points for a given system configuration. By varying the
configuration of the system, one can then vary the amount of resources given to
applications in VMs, track how they respond to the addition or removal of 
resources, and infer resource needs. Experience with such techniques in
isolation has shown that they can be surprisingly
accurate~\cite{antfarm-usenix06,geiger-asplos06}.  Furthermore, modeling
resource needs across a range of configurations with high accuracy is less
important than discovering where in that range the application experiences an
inflection point (e.g., cannot make use of further CPU or
memory)~\cite{sharp:sosp03,informed-sosp95}.

\section{Conclusion}
\label{sec.conclusion}

In this paper we have proposed a novel approach for allocating resources among
competing jobs, relying on Virtual Machine technology and on the optimization of
a well-defined metric that captures notions of performance and of fairness. We
have given a formal definition of a base problem, have proposed several
algorithms to solve it, and have evaluated these algorithms in simulation. We
have identified a promising algorithm that runs quickly, is on par with or
better than its competitors, and is close to optimal in terms of our objective
function. We have then discussed several extensions to our approach to solve
more general problems, namely when jobs are parallel, when the workload is
dynamic, when job resource needs are composite, and when job resource needs are
unknown.  

Future directions include the development of algorithms to solve the resource
allocation adaptation problem, and of strategies for estimating job resource
needs accurately. Our ultimate goal is to develop a new resource allocator as
part of the Usher system~\cite{Usher}, so that our algorithms and techniques can
be used as part of a practical system and evaluated in practical settings.

\bibliographystyle{ieee}
\bibliography{biblio}

\end{document}